\documentclass[]{interact}

\usepackage{epstopdf}
\usepackage[caption=false]{subfig}
\usepackage{mathtools}
\DeclarePairedDelimiter\ceil{\lceil}{\rceil}
\DeclarePairedDelimiter\floor{\lfloor}{\rfloor}

\usepackage{algorithm}
\usepackage{algpseudocode}
\algtext*{EndIf}
\algtext*{EndProcedure}
\algtext*{EndFunction}
\algtext*{EndFor}
\algtext*{EndWhile}

\usepackage{diagbox}
\usepackage{url}

\usepackage[numbers,sort&compress]{natbib}
\bibpunct[, ]{[}{]}{,}{n}{,}{,}
\renewcommand\bibfont{\fontsize{10}{12}\selectfont}
\makeatletter
\def\NAT@def@citea{\def\@citea{\NAT@separator}}
\makeatother

\theoremstyle{plain}
\newtheorem{theorem}{Theorem}[section]

\theoremstyle{definition}
\newtheorem{definition}[theorem]{Definition}

\theoremstyle{remark}
\newtheorem{remark}{Remark}

\begin{document}

\articletype{ARTICLE TEMPLATE}

\title{Four-splitting based coarse-grained multicomputer parallel algorithm for the optimal binary search tree problem}

\author{
\name{Jerry Lacmou Zeutouo\textsuperscript{a}, Vianney Kengne Tchendji\textsuperscript{a}\thanks{CONTACT Vianney Kengne Tchendji. Email: vianneykengne@yahoo.fr}, and Jean Fr\'ed\'eric Myoupo\textsuperscript{b}}
\affil{\textsuperscript{a}Department of Mathematics and Computer Science, University of Dschang, Dschang, Cameroon; \textsuperscript{b}Computer Science Lab-MIS, University of Picardie Jules Verne, Amiens, France}
}

\maketitle

\begin{abstract}
This paper presents a parallel solution based on the coarse-grained multicomputer (CGM) model using the four-splitting technique to solve the optimal binary search tree problem. The well-known sequential algorithm of Knuth solves this problem in $\mathcal{O}\left(n^2\right)$ time and space, where $n$ is the number of keys used to build the optimal binary search tree. To parallelize this algorithm on the CGM model, the irregular partitioning technique, consisting in subdividing the dependency graph into subgraphs (or blocks) of variable size, has been proposed to tackle the trade-off of minimizing the number of communication rounds and balancing the load of processors. This technique however induces a high latency time of processors (which accounts for most of the global communication time) because varying the blocks' sizes does not enable them to start evaluating some blocks as soon as the data they need are available. The four-splitting technique proposed in this paper solves this shortcoming by evaluating a block as a sequence of computation and communication steps of four subblocks. This CGM-based parallel solution requires $\mathcal{O}\left(n^2/\sqrt{p} \right)$ execution time with $\mathcal{O}\left( k \sqrt{p}\right)$ communication rounds, where $p$ is the number of processors and $k$ is the number of times the size of blocks is subdivided. An experimental study conducted to evaluate the performance of this CGM-based parallel solution showed that compared to the solution based on the irregular partitioning technique where the speedup factor is up to $\times$10.39 on one hundred and twenty-eight processors with 40960 keys when $k = 2$, the speedup factor of this solution is up to $\times$13.12 and rises up to $\times$14.93 when $k = 5$.
\end{abstract}


\section{Introduction}
The optimal binary search tree problem consists in determining the binary search tree that will have the lowest overall cost of searching from a set of sorted keys knowing their access probabilities. It is a non-serial polyadic dynamic-programming problem which is characterized by a strong dependency between subproblems \cite{Wah1985}. This problem is widely studied since applications of optimal search binary trees are numerous. One of the most obvious is the search for a word in a dictionary. Indeed, from the words of a dictionary and the access frequency of each word, an optimal binary search tree can be built to quickly and efficiently answer a query. An application derived from the latter is the translation of a word from one language to another \cite{Cormen2009}. El-Qawasmeh \cite{El-Qawasmeh2004} used an optimal binary search tree to solve the word prediction problem. This problem attempts to guess and update the next word in a sentence as it is typed. In view of these different applications, an optimal binary search tree should be constructed from Knuth's sequential algorithm running in $\mathcal{O}\left(n^2\right)$ time and space \cite{Knuth1971}. It is an improved version of Godbole's sequential algorithm requiring $\mathcal{O}\left(n^3\right)$ time and $\mathcal{O}\left(n^2\right)$ space \cite{Godbole1973}.

\subsection{Related work}

In the literature, researchers proposed several parallel solutions of Knuth's sequential algorithm on different parallel computing models. On the PRAM model, Karpinski et al. \cite{Karpinski1994} designed a sublinear time parallel algorithm to construct optimal binary search trees. It requires $\mathcal{O}\left(n^{1-\epsilon}\log n\right)$ computation time with the total work $\mathcal{O}\left(n^{2+2\epsilon}\right)$ for an arbitrarily small constant $0 < \epsilon \leq 0.5$. Recall that the total work is equal to the computation time multiplied by the number of processors \citep{Karpinski1994}. On CREW-PRAM machines, Karpinski et al. \cite{Karpinski1996} presented an algorithm running in $\mathcal{O}\left(n^{0.6}\right)$ computation time with $n$ processors. On realistic models of parallel machines, Craus \cite{Craus2002} proposed an $\mathcal{O}\left(n\right)$-time parallel algorithm using $\mathcal{O}\left(n^2\right)$ processors on very-large-scale integration architectures. Wani and Admad \cite{Wani2019} proposed a parallel implementation of Knuth's sequential algorithm on GPU architectures. With 16384 keys, they achieved a speedup factor of 409 on a NVIDIA GeForce GTX 570 chip and a speedup factor of 745 on a NVIDIA GeForce GTX 1060 chip. On shared-memory architectures, after demonstrating that dependencies of subproblems available in the code implementing the Knuth sequential algorithm enable to generate only 2D tiled code, using the polyhedral model, Bielecki et al. \cite{Bielecki2021} proposed a way of transforming this algorithm to a modified one exposing dependencies to generate 3D parallel tiled code. Experimentations conducted using OpenMP demonstrated that the 3D tiled code considerably outperforms the 2D tiled code. On distributed-memory architectures, many researchers proposed their parallel solutions on the coarse-grained multicomputer (CGM) model.

The CGM model introduced in \citep{Dehne1993,Dehne2002} is the most suitable to design parallel algorithms that are not too dependent on a specific architecture than the systolic and hypercube models. It enables to formalize in just two parameters the performance of a parallel algorithm : the input data size $n$ and the number of processors $p$. A CGM-based parallel algorithm consists in successively repeating a computation round and a communication round until the problem is solved. Processors perform local computations on their data using the best sequential algorithm in each computation round and exchange data through the network in each communication round. All information sent from one processor to another is wrapped into a single long message to minimize the overall message overhead.

Kechid and Myoupo \cite{Kechid2008} proposed the first CGM-based parallel solution to solve the OBST problem in $\mathcal{O}\left(n^2/p\right)$ execution time with $\mathcal{O}\left(p\right)$ communication rounds. Myoupo and Kengne \cite{Myoupo2014} built a CGM-based parallel solution running in $\mathcal{O}\left(n^2/\sqrt{p}\right)$ execution time with $\mathcal{O}\left(\sqrt{p}\right)$ communication rounds. Later on, Kengne et al. \cite{Kengne2016} noticed that the execution time was related to the load balancing and the number of communication rounds. These criteria depend on the partitioning strategy and the distribution scheme strategy used when designing the CGM-based parallel solution:
\begin{enumerate}
	\item When the dependency graph is subdivided into small-size subgraphs (or blocks) \citep{Kechid2008}, the load difference between processors is small if one processor has one more block than another. However, the number of communication rounds will be high.
	\item When the dependency graph is subdivided into large-size blocks \citep{Myoupo2014}, the number of communication rounds of the corresponding algorithm is reduced since there are few blocks. However, the load of processors will be unbalanced.
\end{enumerate}
By generalizing the ideas of the dependency graph partitioning and distribution scheme introduced in \citep{Kechid2008,Myoupo2014}, Kengne et al. \cite{Kengne2016} proposed a CGM-based parallel solution that gives to the end-user the choice to optimize one criterion according to his own goal. The main drawback of this solution is the conflicting optimization criteria owing to the fact that the end-user cannot optimize more than one criterion. Moreover, these criteria have a significant impact on the latency time of processors, which in turn has an impact on the global communication time. Recall that the global communication time is obtained by adding up the latency time of processors and the effective transfer time of data. Indeed, when the number of communication rounds is high, excessive communication will lead to communication overhead which will deteriorate the global communication time. On the other hand, when the load of processors is unbalanced, the evaluation of a block will take a long time because of its large size. So a processor that is waiting for this block to start or continue its computation will wait longer to receive this block. As a result, whatever the chosen criterion, the global communication time will decrease the performance of CGM-based parallel solutions.

To tackle this trade-off, Kengne and Lacmou \cite{Kengne2019} proposed the irregular partitioning technique of the dynamic graph. It consists in subdividing this graph into blocks of variable size. It ensures that the blocks of the first steps (or diagonals) are of large sizes to minimize the number of communication rounds. Thereafter, it decreases these sizes along the diagonals to increase the number of blocks in these diagonals and enable processors to stay active longer. These blocks are fairly distributed over processors to minimize their idle time and balance the load between them. It requires $\mathcal{O}\left( n^2/\sqrt{p}\right)$ execution time with $\mathcal{O}\left(k\sqrt{p}\right)$ communication rounds, where $k$ is the number of times the size of blocks is subdivided. Experimental results showed that the irregular partitioning technique significantly reduced the global communication time compared to the regular partitioning technique proposed in \citep{Kechid2008,Kengne2016,Myoupo2014}.

Nevertheless, this technique also induces a high latency time of processors since it does not enable processors to start evaluating small-size blocks as soon as the data they need are available. Yet, these data are usually available before the end of the evaluation of large-size blocks.

\subsection{Our contribution}
To solve the minimum cost parenthesizing problem (MPP), which is part of the same class of non-serial polyadic dynamic-programming problems than the OBST problem, Lacmou and Kengne \cite{Lacmou2021} proposed the $k$-block splitting technique to reduce this latency time. This technique consists in splitting the large-size blocks into a set of smaller-size blocks called \textit{$k$-blocks}. Thus, to enable processors to start the evaluation of $k$-blocks as soon as possible, a single processor evaluates a block by computing and communicating each $k$-block contained in this block.

The purpose of splitting the blocks into a set of $k$-blocks is to progressively evaluate and communicate them during the evaluation of blocks. Indeed, evaluating a block belonging to the diagonal $d$ in a progressive fashion consists in starting at the diagonal $\ceil*{d/2}$. The $k$-block splitting technique is essentially based on the progressive evaluation of blocks. It is therefore adequate to solve the MPP because Godbole's sequential algorithm gives the possibility to evaluate the nodes of blocks in this way. In contrast, Knuth's sequential algorithm does not enable to perform this kind of evaluation to solve the OBST problem because it does not enable to know in advance when to start or continue the evaluation of a node. Moreover, since the $k$-block are numerous and small when $k$ increases, computing and communicating them to processors that need them will lead to communication overhead \citep{Lacmou2021}. Thus, it would not be meaningful and practical to apply the $k$-block splitting technique to solve the OBST problem because a lot of unnecessary computations could be performed and lead to poor performance.

In this paper, we introduce the four-splitting technique to reduce the latency time of processors caused by the irregular partitioning technique by splitting the large-size blocks into four small-size blocks (or subblocks). Hence, evaluating a block by a single processor will consist of computing and communicating each subblock contained in this block. The goal is the same than the $k$-block splitting technique, that is, to enable processors to start the evaluation of blocks as soon as possible. This CGM-based parallel solution requires $\mathcal{O}\left(n^2/\sqrt{p} \right)$ execution time with $\mathcal{O}\left( k \sqrt{p}\right)$ communication rounds. An experimental study conducted to evaluate the performance of this CGM-based parallel solution showed that compared to the solution based on the irregular partitioning technique where the speedup factor is up to $\times$10.39 on one hundred and twenty-eight processors with 40960 keys when $k = 2$, the speedup factor of this new solution is up to $\times$13.12 and increases up to $\times$14.93 when $k = 5$. We also conducted experiments to determine the ideal number of times the size of the blocks is subdivided. The results showed that this number depends on the architecture where the solution is executed.

This paper is organized as follows: Section 2 describes the OBST problem and presents Knuth's sequential algorithm. Section 3 illustrates the dynamic graph model of the OBST problem introduced in \citep{Kengne2019}. Then, Section 4 presents our CGM-based parallel solution using the four-splitting technique. Section 5 outlines the experimental results obtained, and finally Section 6 concludes this work.

\section{Optimal binary search tree problem}

A binary search tree is a binary tree that maintains a set of sorted keys according to the following rule: for each vertex $x$ of the tree, the key of $x$ is larger than all the keys contained in the left subtree of $x$ and it is smaller than all the keys contained in the right subtree of $x$ \citep{Nagaraj1997}. Figure		 \ref{fig:example_binary_search_tree} shows an example of a binary search tree corresponding to the set of keys $\langle $c, e, f, g, h, k, l, n, o, r, s$\rangle$ sorted in alphabetical order. A binary search tree is typically used to efficiently search for a particular key among a set of sorted keys. The cost of searching for this key is equal to length of the path from the root to this key plus one, which is in fact, the depth of this key \citep{Cormen2009}. For example, in Figure \ref{fig:example_binary_search_tree}, the cost of searching for the key "n" (the root vertex) is equal to 1 and the cost of searching for the key "h" is equal to 5.

\begin{figure}[!t]
	\centering
	\includegraphics[width=\columnwidth]{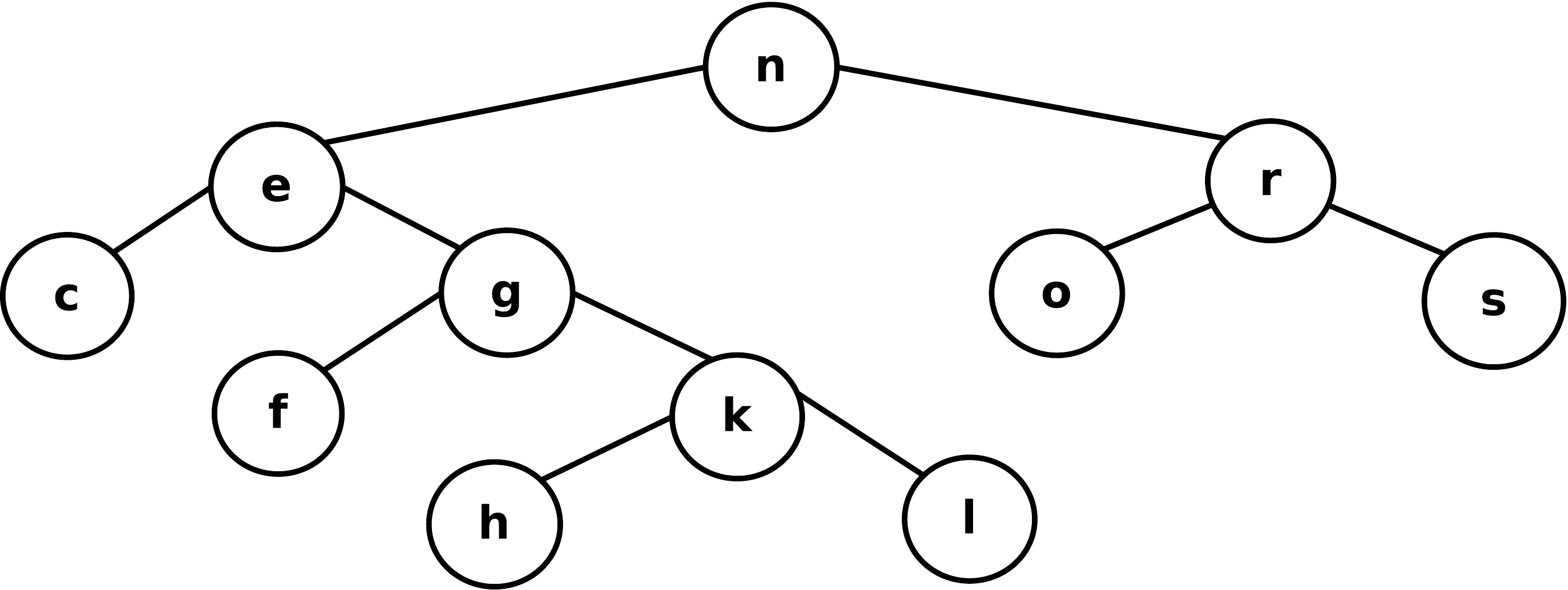}
	\caption{Example of a binary search tree corresponding to the set of letters $\langle$c, e, f, g, h, k, l, n, o, r, s$\rangle$ sorted in alphabetical order}\label{fig:example_binary_search_tree}
\end{figure}

The overall cost of searching, that is equal to the sum of the cost of searching of each key, must be usually as less as possible in many applications. Self-balancing binary search trees, which are binary search trees that automatically attempt to keep its height (maximal number of levels below the root vertex) as small as possible at all times \citep{Knuth1998}, minimize this cost by ensuring that all keys are as near the root vertex as possible. However, when a query frequency is associated with each key, this tree will not be the most efficient in some cases. Indeed, it would be better to place the most frequently queried keys closer to the root vertex, and to place the less frequently queried keys further from the root vertex. For example, consider the set of sorted keys $\langle$a, b, c$\rangle$ with respective frequencies of 3, 1, and 7. The cost of searching for a key $x$ is equal to the depth of $x$ multiplied by the frequency of $x$. Of the five possible binary search trees, depicted in Figure \ref{fig:five_binary_search_trees}, that can be obtained from these keys, the tree that minimizes the overall cost of searching is determined as follows:

\begin{enumerate}
	\item for the first tree, it is equal to $(1\times3) + (2\times7) + (3\times1) = 20$;
	\item for the second tree, it is equal to $(1\times3) + (2\times1) + (3\times7) = 26$;
	\item for the third tree, it is equal to $(1\times1) + (2\times3) + (2\times7) = 21$;
	\item for the fourth tree, it is equal to $(1\times7) + (2\times1) + (3\times3) = 18$;
	\item for the fifth tree, it is equal to $(1\times7) + (2\times3) + (3\times1) = 16$.
\end{enumerate}

\begin{figure}[!t]
	\centering
	\includegraphics[width=\columnwidth]{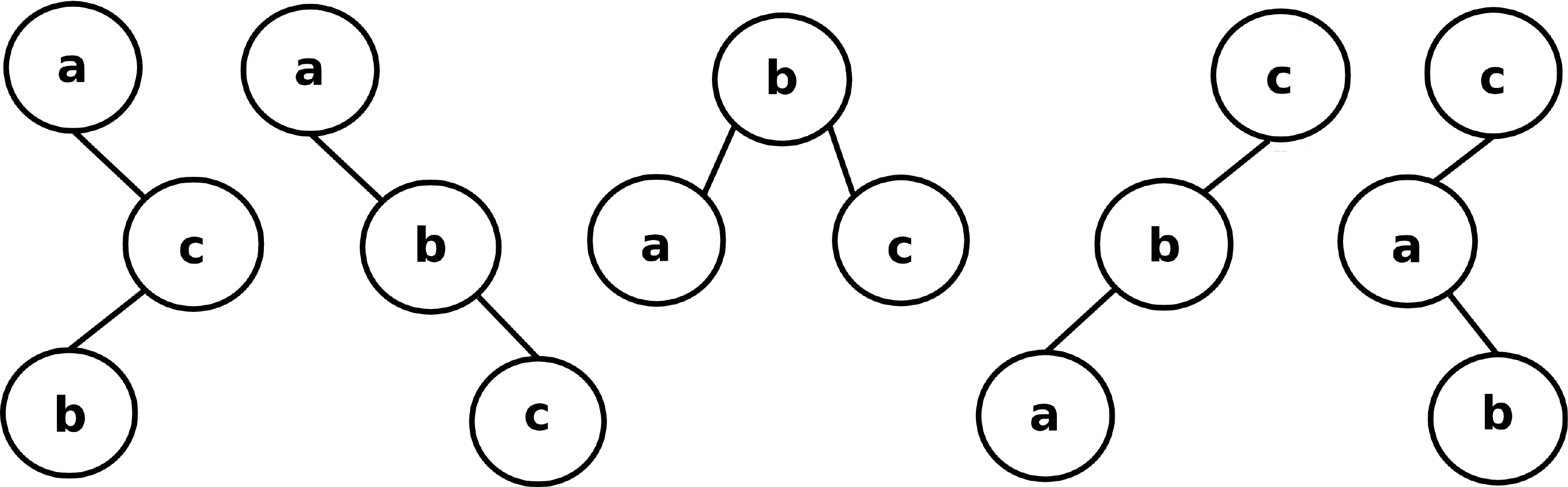}
	\caption{Five possible binary search
		trees obtained from the set of sorted keys $\langle$a, b, c$\rangle$}\label{fig:five_binary_search_trees}
\end{figure}

As shown in Figure \ref{fig:five_binary_search_trees}, the third binary search tree is balanced but the fifth is not. Nevertheless, the overall search cost of the latter is lower than the self-balanced binary search tree (the third tree in Figure \ref{fig:five_binary_search_trees}) and is the lowest of all; thus, it is the optimal binary search tree for the above example.

More formally, the optimal binary search tree (OBST) problem can be defined as follows: consider a set of $n$ sorted keys $K = \langle k_1, k_2, \ldots, k_n\rangle$, such that $k_1 < k_2 < \cdots < k_n$. The probability of finding a key $k_i$ is denoted by $p_i$. In the case where the searched element is not in $K$, consider a set of $n+1$ dummy keys $D = \langle d_0, d_1, \ldots, d_n\rangle$. Indeed, $d_0$ represents the set of values that are smaller than $k_1$ and $d_n$ the set of values that are larger than $k_n$. A dummy key $d_i$, such that $1 \leq i < n$, represents the set of values between $k_i$ and $k_{i+1}$. The probability of finding a dummy key $d_i$ is denoted by $q_i$. In summary, either the search ends in success (i.e. finding the searched key $k_i$) or in failure (i.e. finding a dummy key $d_i$); thus the sum of probabilities of success and failure is equal to 1. A binary search tree $T$ constructed from these different keys is composed of $n$ internal vertices in the set $K$  and $n+1$ leaves in the set $D$. The overall cost of searching of $T$ is given by Equation (\ref{eq:overall_cost_searching_tree}) :
\begin{equation}\label{eq:overall_cost_searching_tree}
	Cost(T) = \sum_{i = 1}^{n} \left(depth(k_i) \times p_i\right) + \sum_{i = 0}^{n} \left(depth(d_i) \times q_i\right)
\end{equation}
where $depth(x)$ is the depth of the key $x$. The OBST problem consists in finding the binary search tree that will have the lowest overall cost of searching.

Let $w(i,j)= p_{i+1}+ \cdots + p_j + q_{i} + q_{i+1}+\cdots+q_j$. The minimum cost of searching of a tree $T_{i,j}$ from the set of keys $K_{i,j} = \langle k_{i+1}, k_{i+2}, \ldots, k_j\rangle$ and the set of dummy keys $D_{i,j} = \langle d_{i}, d_{i+1}, \ldots, d_j\rangle$ is denoted by $Tree[i,j]$ and defined by:
\begin{equation}\label{eq:obst}
Tree[i,j]=\begin{cases}
		q_i                                                                                    & \text{if $ 0 \leq i = j \leq n $,} \\
		\min\limits_{i \leq k < j} \left\lbrace  Tree[i,k] + Tree[k+1,j] + w(i,j)\right\rbrace & \text{if $0 \leq i < j \leq n$.}
	\end{cases}
\end{equation}
Finding the optimal binary search tree from the set of keys $K$ and the set of dummy keys $D$ is reduced to compute $Tree[0,n]$. The straightforward sequential algorithm of Godboble \cite{Godbole1973}, running in $\mathcal{O}\left(n^3\right)$ time and $\mathcal{O}\left(n^2\right)$ space, is given by Algorithm \ref{alg:obst_sequential_algorithm}. The dynamic-programming (DP) table (named $Tree$ in Algorithm \ref{alg:obst_sequential_algorithm}) stores the value of the optimal cost of searching of $T_{i,j}$ (see line \ref{alg:obst_sequential_algorithm:save_in_dp_table:state} in Algorithm \ref{alg:obst_sequential_algorithm}). The tracking table (named $Cut$ in Algorithm \ref{alg:obst_sequential_algorithm}) stores the value of the index $k$ which minimizes $Tree[i,j]$ (see line \ref{alg:obst_sequential_algorithm:save_in_track_table:state} in Algorithm \ref{alg:obst_sequential_algorithm}). It is the optimal decomposition value of $T_{i,j}$ into two subtrees. For a problem of size $n=3$, the dependency graph partitioning and the DP table are illustrated in Figures \ref{fig:obst_problem_n3_dp_dependency} and \ref{fig:obst_problem_n3_dp_table}, respectively.

\alglanguage{pseudocode}
\begin{algorithm}[!t]
	\caption{Godbole's sequential algorithm} \label{alg:obst_sequential_algorithm}
	\begin{algorithmic}[1]
		\For{$i = 0$ to $n$} \label{alg:obst_sequential_algorithm:init:begin}
		\State $Tree[i,i] \gets q_i$; \label{alg:obst_sequential_algorithm:init:end}
		\EndFor
		\For{$d = 1$ to $n$}
		\For{$i = 0$ to $n-d+1$}
		\State $j \gets n-d+1$;
		\State $Tree[i,j] \gets \infty$;
		\For{$k = i$ to $j-1$} \label{alg:obst_sequential_algorithm:loop:3:start}
		\State $c \gets Tree[i,k] + Tree[k+1,j] +  w(i,j)$; \label{alg:obst_sequential_algorithm:core_computation:state}
		\If{$c < Tree[i,j]$}
		\State $Tree[i,j] \gets c$; \label{alg:obst_sequential_algorithm:save_in_dp_table:state}
		\State $Cut[i,j] \gets k$; \label{alg:obst_sequential_algorithm:save_in_track_table:state}
		\EndIf
		\EndFor
		\EndFor
		\EndFor
	\end{algorithmic}
\end{algorithm}

\begin{figure}[!t]
	\centering
	\leavevmode
	\subfloat[][Dependency graph]{%
		\label{fig:obst_problem_n3_dp_dependency}
		\includegraphics[width=.53\columnwidth]{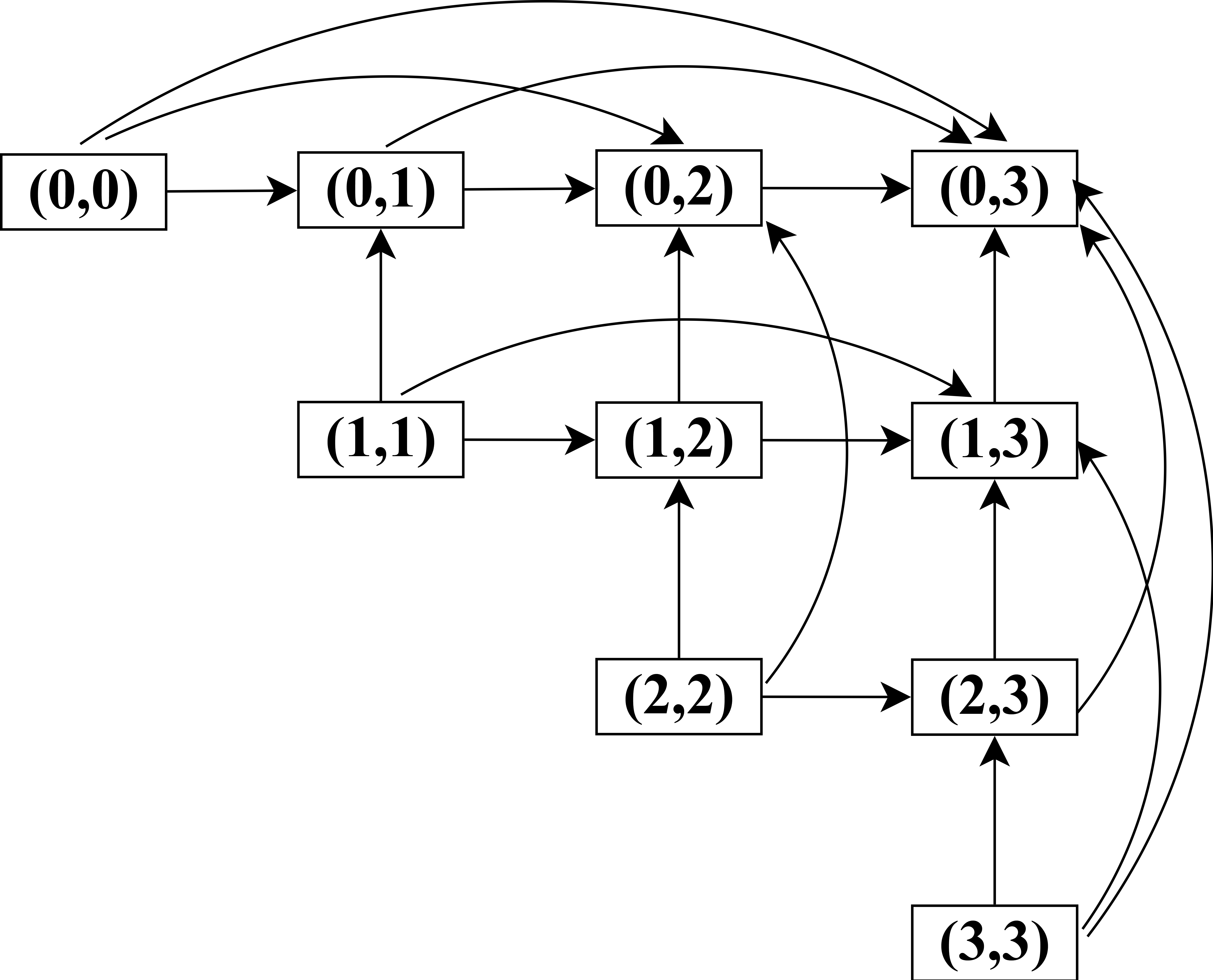}
	}
	\subfloat[][DP table]{%
		\label{fig:obst_problem_n3_dp_table}
		\includegraphics[width=.43\columnwidth]{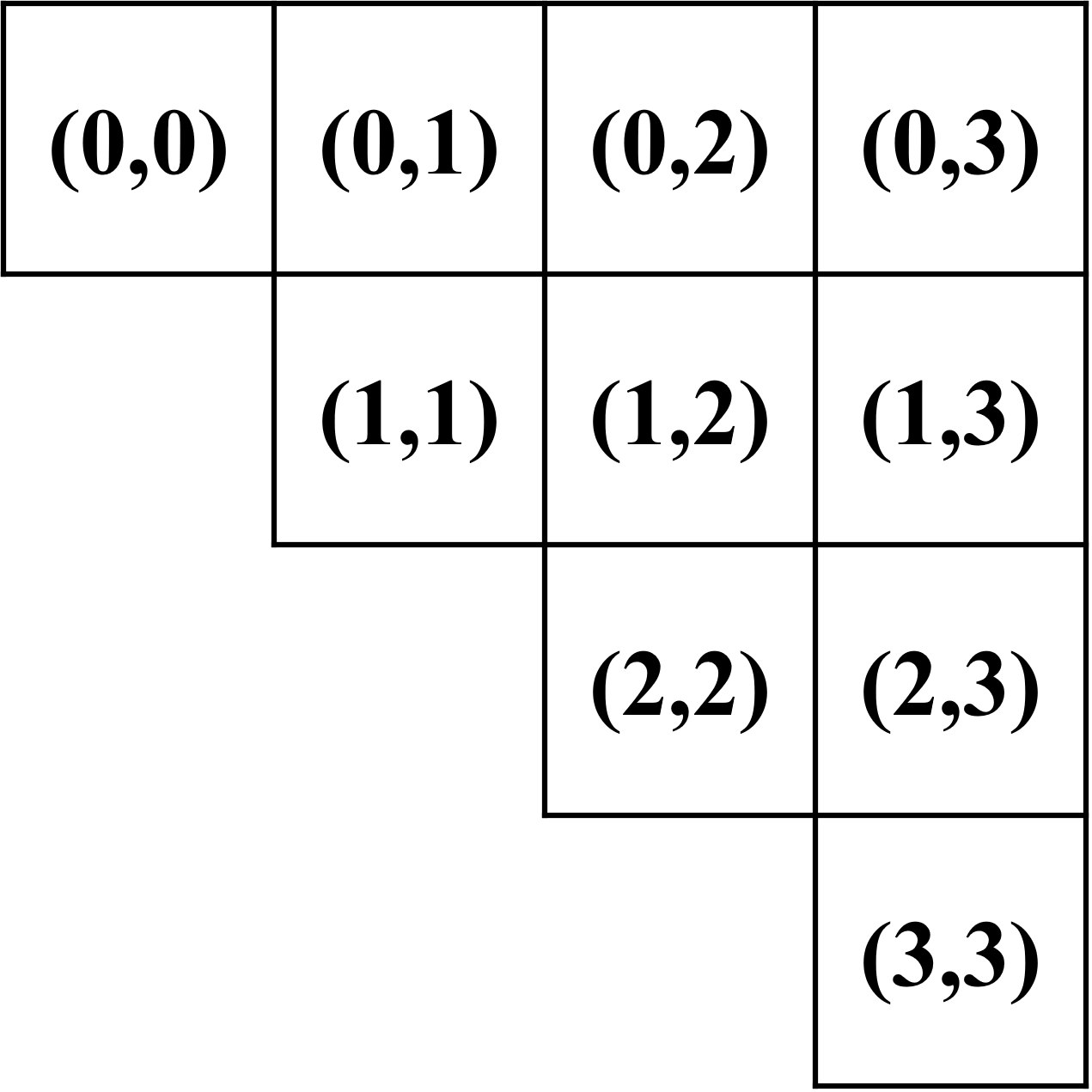}}
	\caption{Dependency graph and dynamic-programming table used to compute $Tree[0, 3]$}
\end{figure}

Knuth \citep{Knuth1971} noticed entries of tracking table $Cut$ satisfy the \textit{monotonicity property} such that $Cut[i,j-1] \leq Cut[i,j] \leq Cut[i+1,j]$ for all $0 \leq i < j \leq n$. This property means that computing $Cut[i,j]$ consists in finding all indices between $Cut[i,j-1]$ and $Cut[i+1,j]$, instead of between $i$ and $j$ as it is done in line \ref{alg:obst_sequential_algorithm:loop:3:start} of Algorithm \ref{alg:obst_sequential_algorithm}. He thus achieves to evaluate the $\Theta\left(n^2\right)$ subproblems in constant time. Algorithm \ref{alg:obst_sequential_algorithm_knuth} draws a big picture.

\alglanguage{pseudocode}
\begin{algorithm}[!t]
	\caption{Knuth's sequential algorithm} \label{alg:obst_sequential_algorithm_knuth}
	\begin{algorithmic}[1]
		\For{$i = 0$ to $n$} \label{alg:obst_sequential_algorithm_knuth:init:begin}
		\State $Tree[i,i] \gets q_i$; \label{alg:obst_sequential_algorithm_knuth:init:end}
		\EndFor
		\For{$d = 1$ to $n$}
		\For{$i = 0$ to $n-d+1$}
		\State $j \gets n-d+1$;
		\State $Tree[i,j] \gets \infty$;
		\For{$k = Cut[i,j-1]$ to $Cut[i+1,j]$} \label{alg:obst_sequential_algorithm_knuth:loop:3:start}
		\State $c \gets Tree[i,k] + Tree[k+1,j] +  w(i,j)$; \label{alg:obst_sequential_algorithm_knuth:core_computation:state}
		\If{$c < Tree[i,j]$}
		\State $Tree[i,j] \gets c$; \label{alg:obst_sequential_algorithm_knuth:save_in_dp_table:state}
		\State $Cut[i,j] \gets k$; \label{alg:obst_sequential_algorithm_knuth:save_in_track_table:state}
		\EndIf
		\EndFor
		\EndFor
		\EndFor
	\end{algorithmic}
\end{algorithm}

One downside to the speedup of Knuth, however, is that the number of comparison operations for entries in the same diagonal varies from one entry to another compared to the classical version \cite{Knuth1971}. Therefore, when designing parallel algorithms based on Knuth's sequential algorithm, there is no guarantee that the processors will have the same load if entries on a diagonal are fairly distributed among them. It is the main parallelization constraint of this algorithm.

\section{Dynamic graph model of the OBST problem}
Kengne and Lacmou \citep{Kengne2019} showed that the OBST problem can be solved through the shortest path algorithms on weighted dependency graph. Indeed, they proposed a graph model called \textit{dynamic graph} that is formulated from Equation (\ref{eq:obst}). For a problem of size $n$, this graph is denoted by $D_n$ and defined as follows :

\begin{definition}
	Given a set of $n$ sorted keys, a dynamic graph $D_n = (V,E \cup E')$ is defined as a set of vertices,
	\[
		\begin{split}
			V & = \{(i, j) : 0 \leq i \leq j \leq n\} \cup \{(-1,-1)\}\\
		\end{split}
	\]
	a set of unit edges,
	\[
		\begin{split}
			E & = \{(i, j) \rightarrow (i, j + 1) : 0 \leq i \leq j < n\} \cup \{(i, j) \uparrow (i - 1, j) : \\
			& \hspace{1.3em} 0 < i \leq j \leq n\} \cup \{(-1,-1) \nearrow (i, i) : 0 \leq i \leq n\}
		\end{split}
	\]
	a set of jumps,
	\[
		\begin{split}
			E' & = \{(i, j) \Rightarrow (i, t) : 0 \leq i < j < t \leq n\} \cup \{(s, t) \Uparrow (i, t) : \\
			& \hspace{1.3em} 0 \leq i < s < t \leq n\}
		\end{split}
	\]
	a weight function $W$ such that
	\[
		\begin{array}{l l}
			W((i, j) \rightarrow (i, j+1)) = q_{j+1} + w(i,j+1) & 0 \leq i \leq j < n        \\
			W((i, j) \uparrow (i - 1, j)) = q_{i-1} + w(i-1,j)  & 0 < i \leq j \leq n        \\
			W((-1,-1) \nearrow (i, i))   = q_i                  & 0 \leq i \leq n            \\
			W((i, k) \Rightarrow (i, j))  = SP[k+1,j] + w(i,j)  & 0 \leq i < k < j \leq n    \\
			W((k+1, j) \Uparrow (i, j))  = SP[i,k] + w(i,j)     & 0 \leq i \leq k < j \leq n
		\end{array}
	\]
\end{definition}

A square matrix of size $n$ called \textit{shortest path matrix}, and denoted by $SP$, is used to store in the cell $SP[i,j]$ the shortest path from node $(-1, -1)$ to $(i,j)$. They showed that the computation of $Tree[i,j]$ is equivalent to search in $D_n$ the shortest path from node $(-1, -1)$ to $(i,j)$. Therefore, it is straightforward to prove that Knuth's sequential algorithm is equivalent to compute the shortest paths from $(-1, -1)$ to the other vertices in a dynamic graph $D_n$, incrementally, diagonal after diagonal, from left to right. Thus, the shortest path corresponds to the optimal binary search tree. Figure \ref{fig:Dn_n3} shows a dynamic graph $D_3$ for a problem of size $n = 3$. It has the same form as the dependency graph between subproblems depicted in Figure \ref{fig:obst_problem_n3_dp_dependency}, with an additional node $(-1, -1)$.

\begin{figure}[!t]
	\centering
	\leavevmode
	\subfloat[][Dynamic graph $D_3$]{%
		\label{fig:Dn_n3}
		\includegraphics[width=.5025\columnwidth]{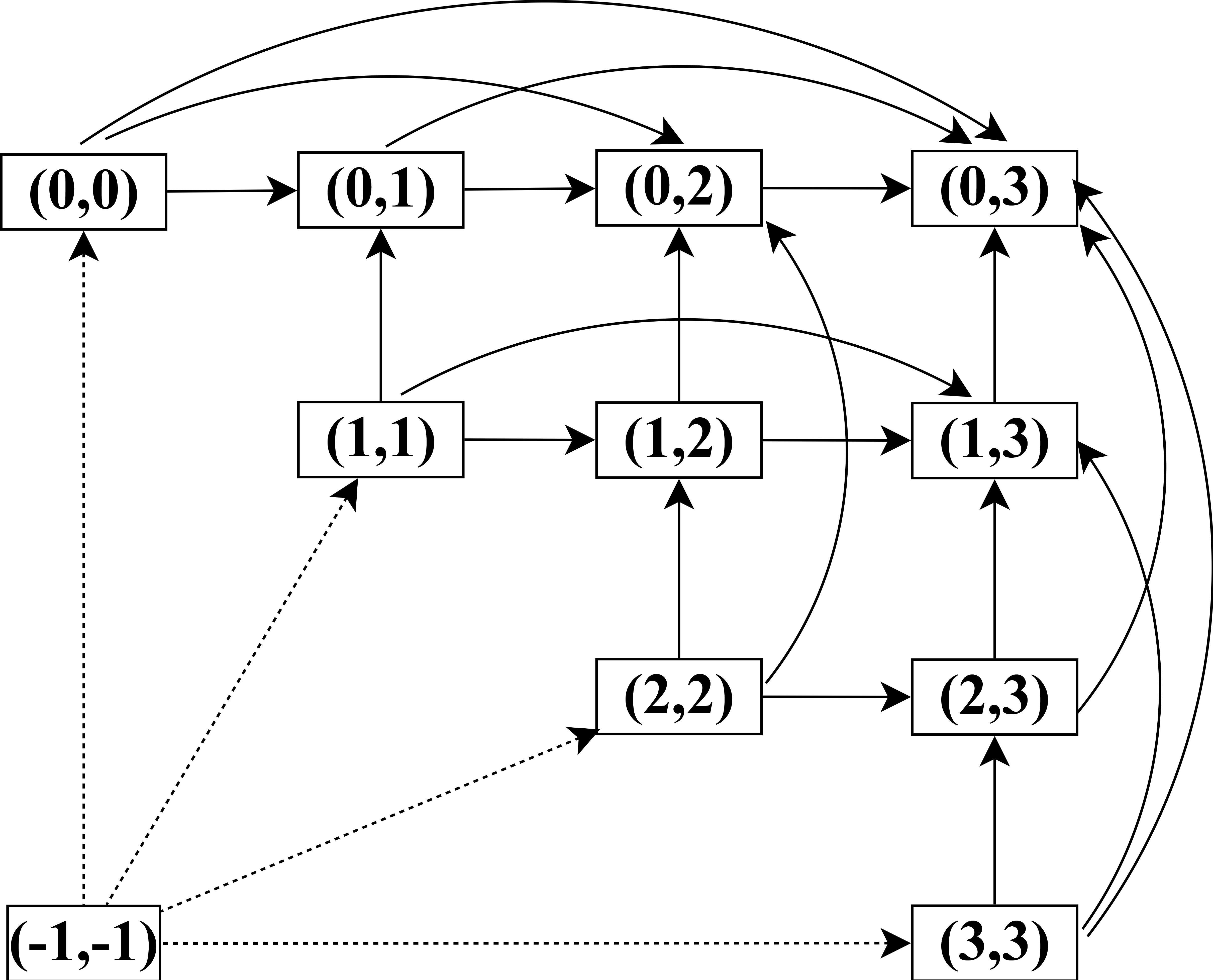}
	}
	\subfloat[][Dynamic graph ${D}_{3}^{\prime}$]{%
		\label{fig:Dn_n3_without_vert_and_hor_arc}
		\includegraphics[width=.4475\columnwidth]{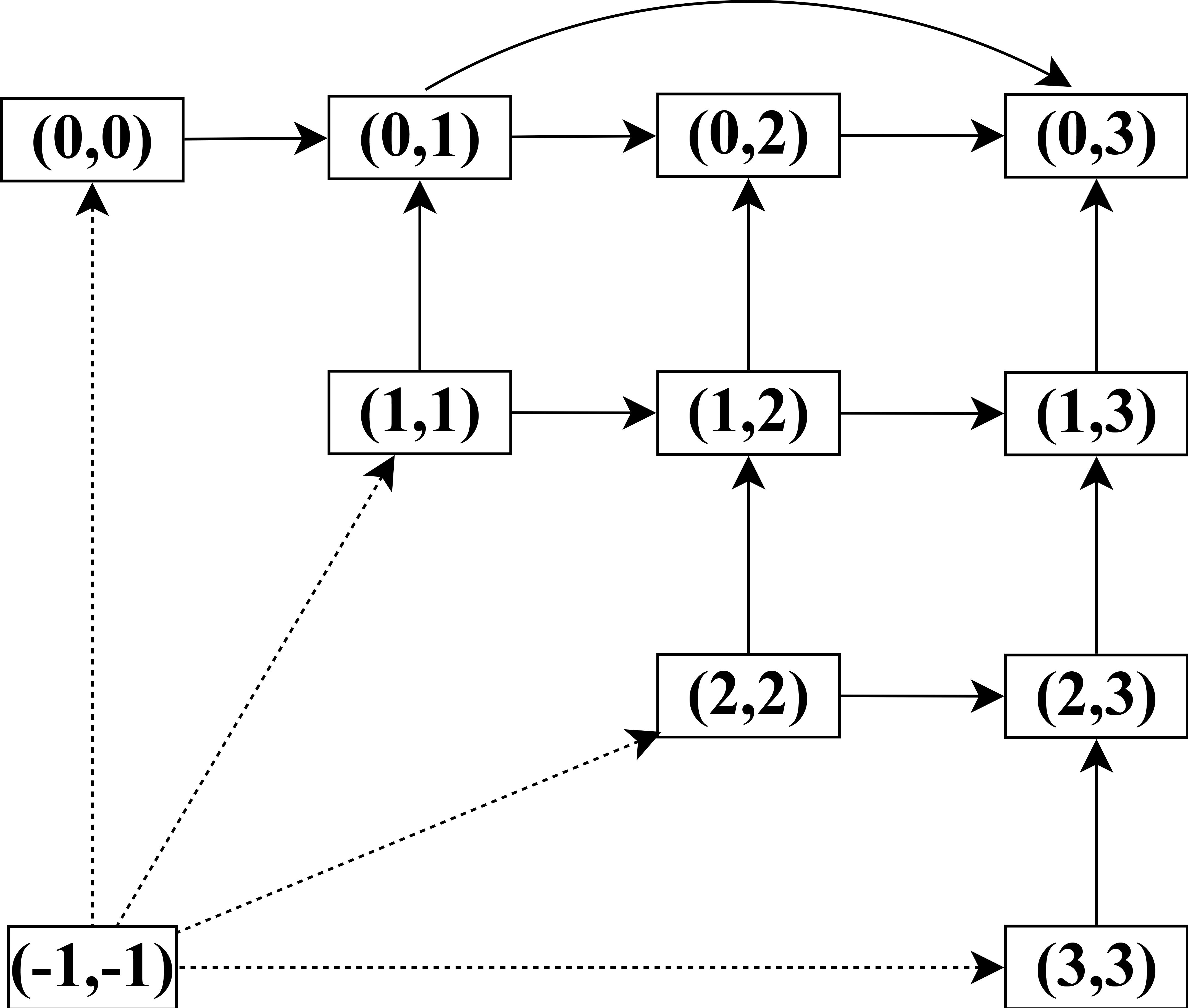}}
	\caption{Dynamic graphs $D_3$ and ${D}_{3}^{\prime}$ for $n=3$}
\end{figure}

Given a problem of size $n$ and its corresponding dynamic graph $D_n$ :

\begin{theorem}[Duality Theorem]\label{theo:duality}
	If a shortest path from $(-1,-1)$ to $(i, j)$ contains the jump $(i, k) \Rightarrow (i, j)$, then there is a dual shortest path containing the unit edge $(k+1, j) \uparrow (i, j)$ when $0 \leq i = k < j \leq n$ and the jump $(k+1, j) \Uparrow (i, j)$ when $0 \leq i < k < j \leq n$.
\end{theorem}

Theorem \ref{theo:duality} is fundamental because it makes it possible to avoid redundant computations when looking for the value of the shortest path of $D_n$'s vertices. Indeed, for any vertex $(i, j)$, among all its shortest paths containing jumps, only those that contain only horizontal jumps are evaluated. The input graph of our CGM-based parallel solution is therefore a subgraph of $D_n$ denoted by ${D}_{n}^{\prime}$, in which the set of edges from $(i, k)$ to $(i, j)$, such that $0 \leq i = k < j \leq n$, and from $(k + 1, j)$ to $(i, j)$, such that $0 \leq i < k < j \leq n$, is removed. Figure \ref{fig:Dn_n3_without_vert_and_hor_arc} shows the dynamic graph ${D}_{3}^{\prime}$.

\section{Our CGM-based parallel solution using the four-splitting technique}

\subsection{Dynamic graph partitioning}
The irregular partitioning technique consists in subdividing the shortest path matrix into submatrices (blocks) of varying size (irregular size) to enable a maximum of processors to remain active longer. The idea is to increase the number of blocks of diagonals whose this number is lower or equal to half of the first one through the block fragmentation technique. This technique aims to reduce the block size by dividing it into four subblocks. To minimize the number of communication rounds, it begins to subdivide the shortest path matrix with large-size blocks from the largest diagonal (the first diagonal of blocks) to the diagonal located just before the one whose number of blocks is half of the first one. Then, since the number of blocks per diagonal quickly becomes smaller than the number of processors, to increase the number of blocks of these diagonals and enable a maximum of processors to remain active, it fragments all the blocks belonging to the next diagonal until the last one to catch up or exceed by one notch the number of blocks of the first diagonal. It reduces the idle time of processors and promotes the load balancing. This process is repeated $k$ times, after which the block sizes are no longer modify, and the rest of the partitioning becomes traditional because an excessive fragmentation would lead to a drastic rise in the number of communication rounds. After performing $k$ fragmentations, a block belonging to $l$th level of fragmentation has been subdivided $l$ times, $0 \leq l \leq k$.

The four-splitting technique consists of splitting the large-size blocks into four small-size blocks (or subblocks) after performing $k$ fragmentations to reduce the latency time of processors. The subblocks of the blocks belonging to the $l$th level of fragmentation must have the same size than the blocks belonging to the $(l+1)$th level of fragmentation. The goal is to enable processors to start the evaluation of blocks as soon as possible. Hence, evaluating a block by a single processor will consist of computing and communicating each subblock contained in this block.

By denoting $f(p)=\ceil*{\sqrt{2p}}$, $\theta(n,p)=\ceil*{(n+1)/f(p)}$, and $\theta(n,p,l)=\ceil*{\theta(n,p)/2^l}$, formally, we subdivide the shortest path matrix $SP$ into blocks (denoted by $SM(i, j)$), and split the large-size blocks into four subblocks. Thus, a block $SM(i, j)$ belonging to the $l$th level of fragmentation, such that $l < k$, is a $\theta(n,p,l) \times \theta(n,p,l)$ matrix and is subdivided into four subblocks of size $\theta(n,p,l+1) \times \theta(n,p,l+1)$. The blocks of the $k$th level of fragmentation are not splitting into four as these are the smallest blocks. Figures \ref{fig:DAG_4s_part_k_1_p_3_n_32}, \ref{fig:DAG_4s_part_k_2_p_3_n_32}, \ref{fig:DAG_4s_part_k_1_p_5_n_32}, and \ref{fig:DAG_4s_part_k_2_p_5_n_32} depict four scenarios of this partitioning for $n = 31$, $k \in \{1,2\}$, and $p \in \{$3, 4, 5, 6, 7, 8$\}$. The number in each block represents the diagonal in which it belongs. 

\begin{figure*}[!t]
	\centering
	\leavevmode
	\subfloat[][$p \in \{$3, 4$\}$ and $k = 1$]{%
		\label{fig:DAG_4s_part_k_1_p_3_n_32}
		\includegraphics[width=.485\textwidth]{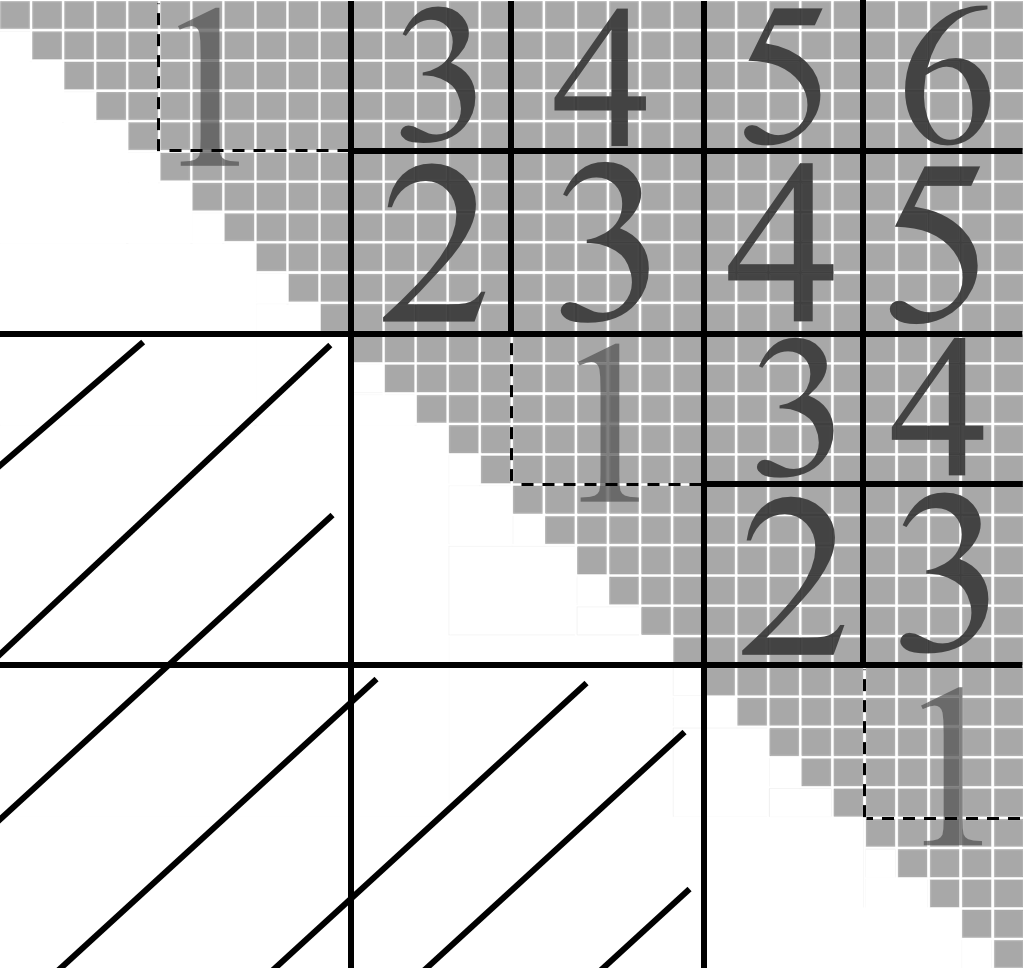}
	}
	\subfloat[][$p \in \{$3, 4$\}$ and $k = 2$]{%
		\label{fig:DAG_4s_part_k_2_p_3_n_32}
		\includegraphics[width=.485\textwidth]{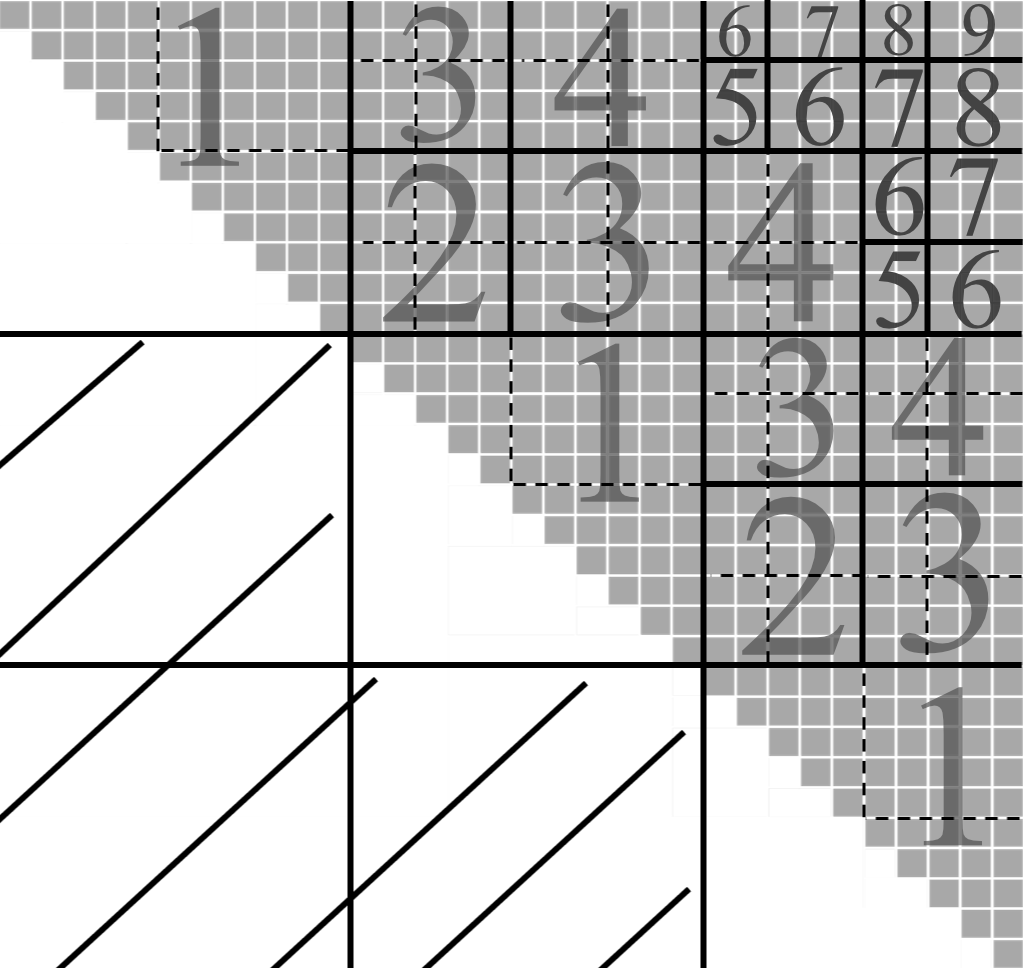}
	}\\
	\subfloat[][$p \in \{$5, 6, 7, 8$\}$ and $k = 1$]{%
		\label{fig:DAG_4s_part_k_1_p_5_n_32}
		\includegraphics[width=.485\textwidth]{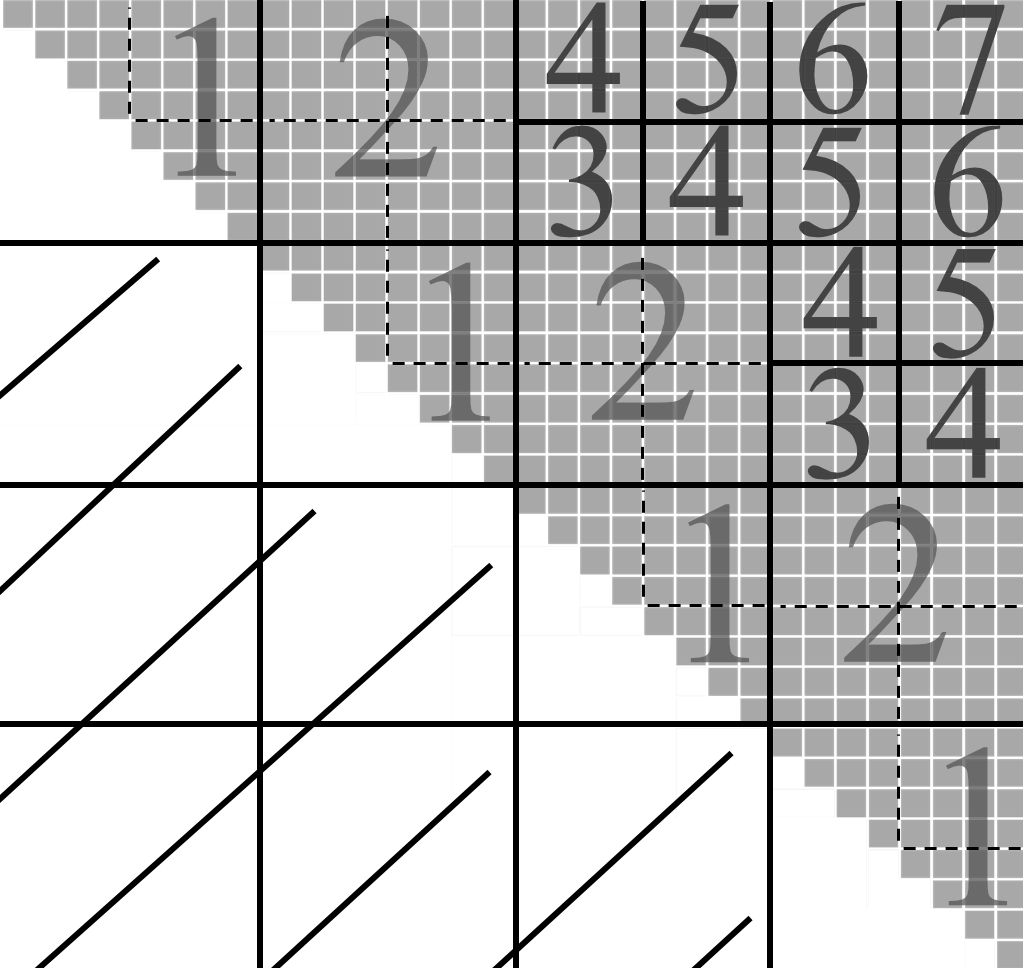}
	}
	\subfloat[][$p \in \{$5, 6, 7, 8$\}$ and $k = 2$]{%
		\label{fig:DAG_4s_part_k_2_p_5_n_32}
		\includegraphics[width=.485\textwidth]{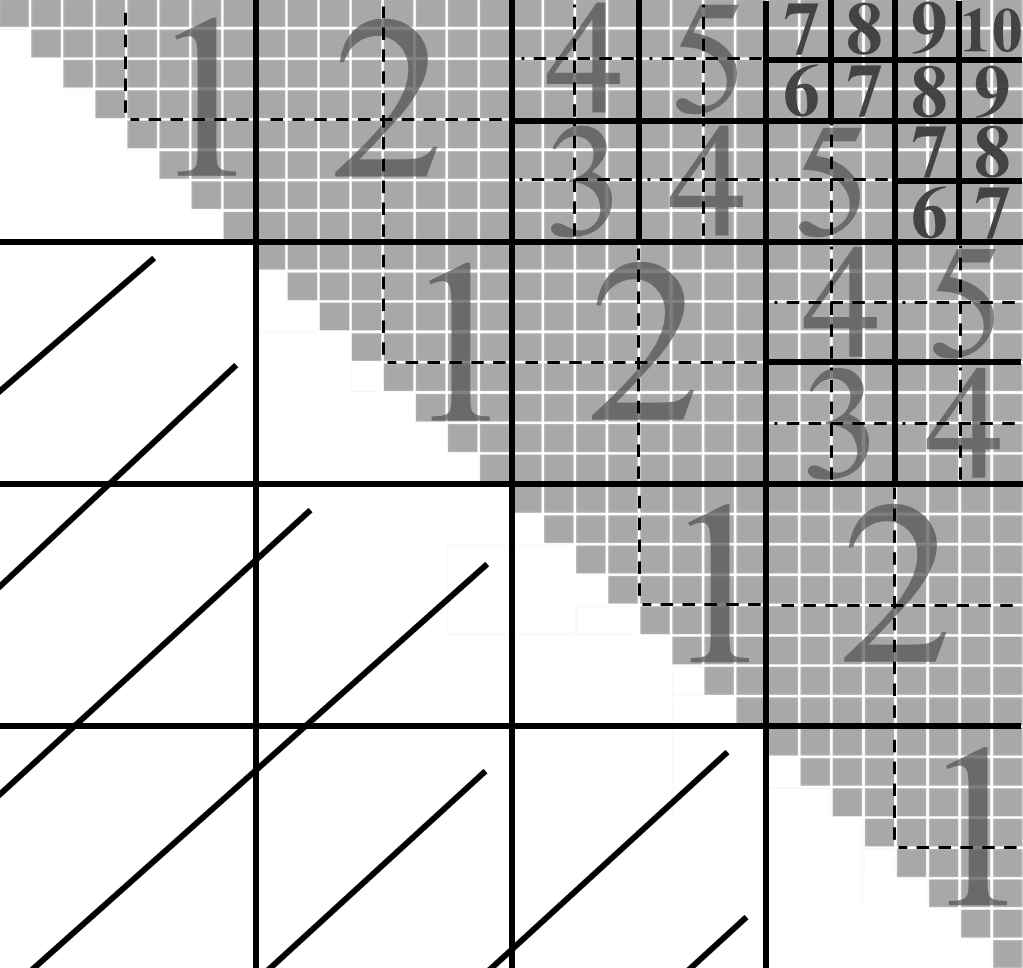}
	}
	\caption[Four-splitting technique of the shortest path matrix for $n = 31$, $k \in \{1,2\}$, and $p \in \{$3, 4, 5, 6, 7, 8$\}$]{Four-splitting technique of the shortest path matrix for $n = 31$, $k \in \{1,2\}$, and $p \in \{$3, 4, 5, 6, 7, 8$\}$. For $p \in \{$3, 4$\}$, $SP$ is partitioned into fifteen blocks and twenty-one subblocks when $k = 1$, and into twenty-four blocks and fifty-seven subblocks when $k = 2$. For $p \in \{$5, 6, 7, 8$\}$, $SP$ is partitioned into nineteen blocks and thirty-six subblocks when $k = 1$, and into twenty-eight blocks and seventy-two subblocks when $k = 2$}
\end{figure*}

\begin{remark} 
	Some relevant points can be noticed about this partitioning :
	\begin{enumerate}
		\item {the blocks of the first diagonal are $\theta(n, p) \times \theta(n, p)$ upper triangular matrices splitting into tree subblocks;}
		\item {all the blocks are not full when $n \bmod \left(2^k\times f(p)\right) \neq 0 $, for example in Figure \ref{fig:DAG_4s_part_k_2_p_3_n_32} where $32 \bmod \left( 2^2\times 3\right)  = 8 \neq 0$);}
		\item {a block belonging to the $l$th level of fragmentation is full if it is a non-triangular matrix of size $\theta(n, p, l) \times \theta(n, p, l)$;}
		\item {one fragmentation increases up to ($\ceil*{f(p)/2} +1$) the number of diagonals;}
		\item {when $f(p)$ is odd, the number of blocks in a diagonal after each fragmentation exceeds by one notch the number of the larger blocks of the first diagonal. This is illustrated in Figure \ref{fig:DAG_4s_part_k_1_p_3_n_32}, where there are three blocks in the first diagonal and four blocks in third diagonal.}
	\end{enumerate}
\end{remark}

\subsection{Blocks' dependency analysis}

\begin{figure}[!t]
	\centering
	\leavevmode
	\subfloat[][Dependencies of $SM(i,j)$]{%
		\includegraphics[width=.585\columnwidth]{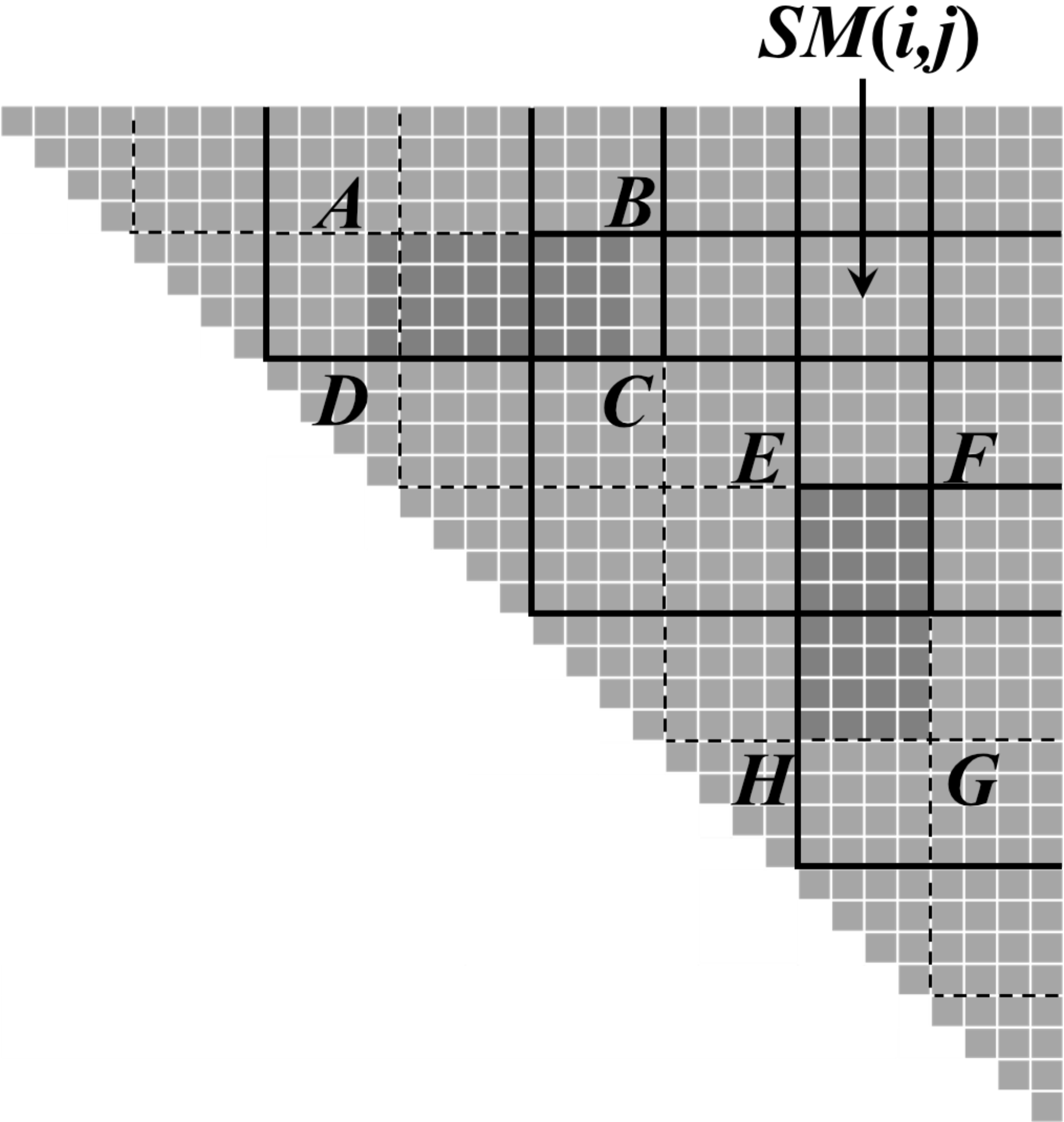}
		\label{fig:DAG_4s_dependency_obst_n_32}
	}
	\subfloat[][Extremities of $SM(i,j)$]{%
		\label{fig:block_delimitation}
		\includegraphics[width=.385\columnwidth]{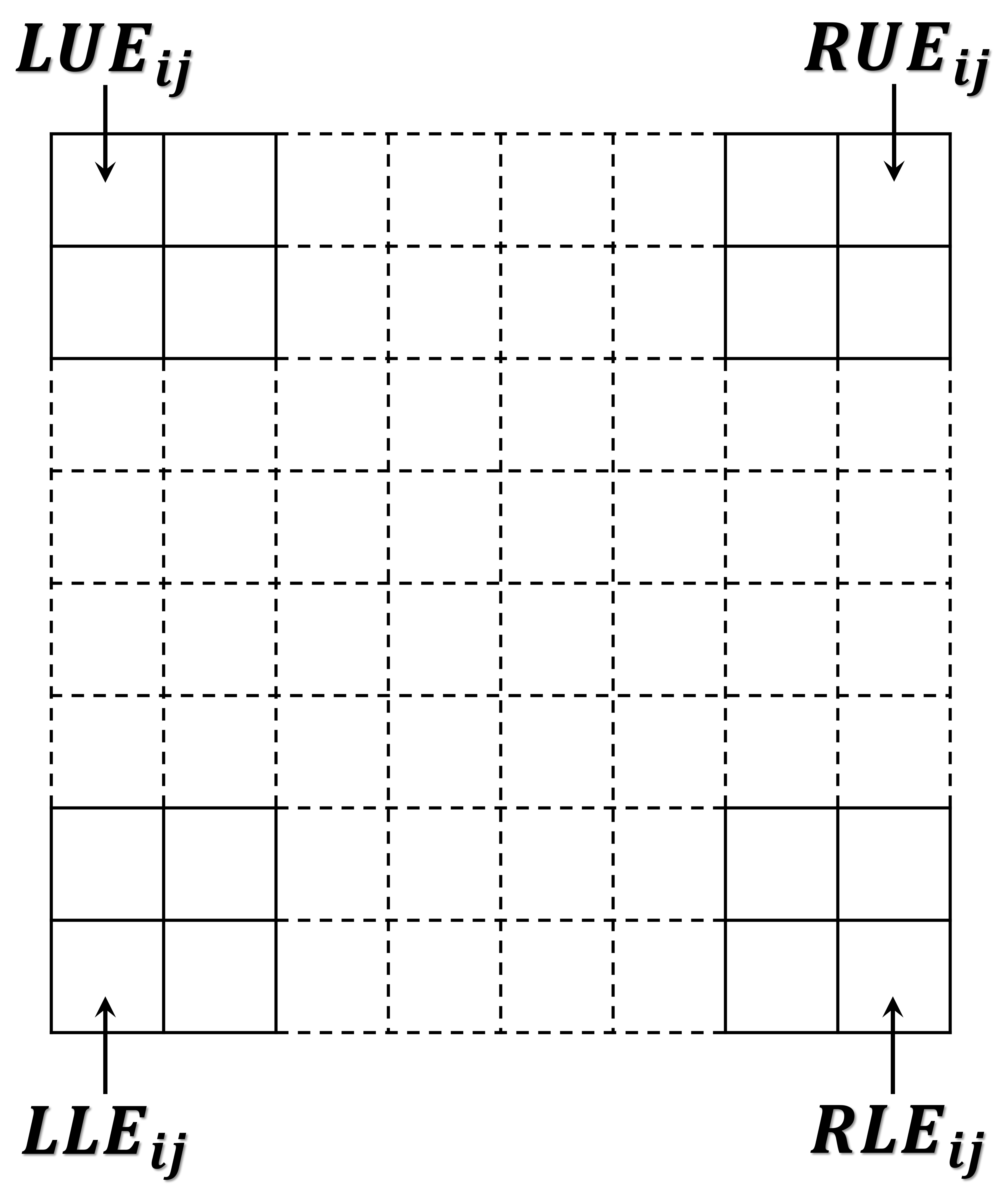}}
	\caption{Dependencies and extremities of a block $SM(i,j)$ after applying the four-splitting technique}
\end{figure}

Figures \ref{fig:DAG_4s_dependency_obst_n_32} and \ref{fig:block_delimitation} show respectively an example of dependencies and extremities of a block $SM(i,j)$ after applying the four-splitting technique. The extremities of $SM(i,j)$ are defined by :

\begin{itemize}
	\item the leftmost upper entry $LUE_{ij} = (i, j - \theta(n,g,l) + 1)$;
	\item the rightmost upper entry $RUE_{ij} = (i,j)$;
	\item the leftmost lower entry $LLE_{ij} = (i + \theta(n,g,l) - 1, j - \theta(n,g,l) + 1)$;
	\item the rightmost lower entry $RLE_{ij} = (i + \theta(n,g,l) - 1,j)$.
\end{itemize}
Figure \ref{fig:DAG_4s_dependency_obst_n_32} depicts eight points ($A$, $B$, $C$, $D$, $E$, $F$, $G$, and $H$) that identify blocks on which the block $SM(i,j)$ depends :

\begin{itemize}
	\item {$A = Tree\left[LUE_{i,j-3\times\theta(n,g,l)-1}\right]$}
	\item {$B = Tree\left[LUE_{i,j-\theta(n,g,l)-2}\right]$}
	\item {$C = Tree\left[LLE_{i,j-\theta(n,g,l)-2}\right]$}
	\item {$D = Tree\left[LLE_{i,j-3\times\theta(n,g,l)-1}\right]$}
	\item {$E = Tree\left[LLE_{i+\theta(n,g,l)+1,j}\right]$}
	\item {$F = Tree\left[RLE_{i+\theta(n,g,l)+1,j}\right]$}
	\item {$G = Tree\left[RLE_{i+3\times\theta(n,g,l),j}\right]$}
	\item {$H = Tree\left[LLE_{i+3\times\theta(n,g,l),j}\right]$}
\end{itemize}
All lower blocks and subblocks (a part of a block) that are in the same row and column than the block $SM(i,j)$ are no longer absolutely required to evaluate $SM(i,j)$ as a consequence of the speedup of Knuth \citep{Knuth1971}. In Figure \ref{fig:DAG_4s_dependency_obst_n_32} for example, only the five most shaded blocks will be needed. In other cases, it could have been just four, three or two blocks. Moreover, it may happen that some nodes of blocks located in extremities are needed instead of the whole blocks. It depends on the input data. Therefore, the speedup of Knuth doesn't enable to estimate the exact load of a block before its evaluation \citep{Myoupo2014}. However, the blocks on which the block $SM(i,j)$ depends aren't located on the same diagonal than $SM(i,j)$. Hence, they can be carried out in parallel.

\subsection{Mapping blocks onto processors}

The mapping used here consists in assigning the blocks of a given diagonal from the upper leftmost corner to the lower rightmost corner. First, the blocks of the first diagonal are assigned from processor $P_0$ to processor $P_{f(p)-1}$. Then, the process is repeated on the next diagonal starting with processor $P_{f(p)}$, and so on until a block has been assigned to each processor. The mapping starts again with processor $P_0$, and continues along the diagonals with a snake-like path until the last diagonal. Figure \ref{fig:DAG_irreg_part_mapping_snake_p_5_n_32} illustrates this mapping on six processors.

\begin{figure}[!t]\centering
	\includegraphics[width=.585\columnwidth]{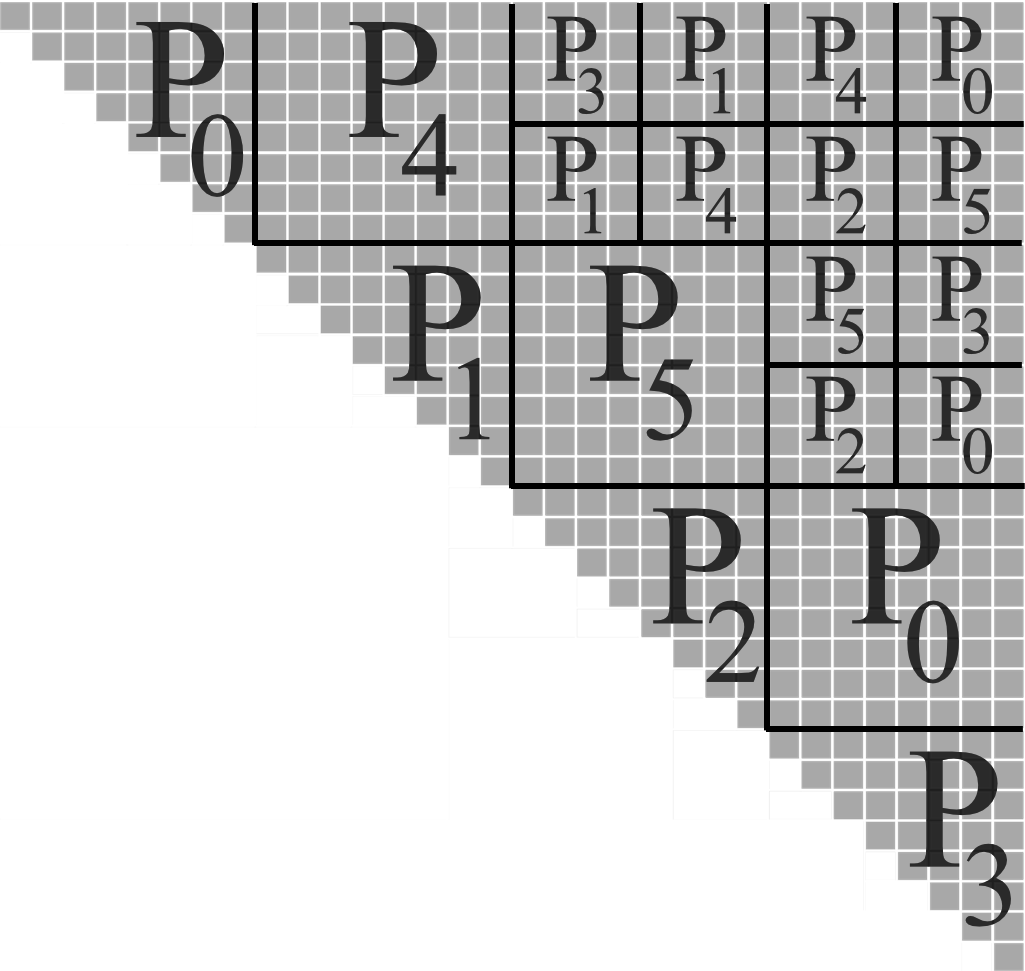}
	\caption{Snake-like mapping on six processors when $k = 1$}
	\label{fig:DAG_irreg_part_mapping_snake_p_5_n_32}
\end{figure}

The communication scheme corresponding to this mapping is simple and easy to implement. Moreover, this mapping enables processors to remain active as soon as possible and minimizes the number of communication rounds. It also ensures load balancing because it enables some processors to evaluate at most one block more than others. For example, Figure \ref{fig:DAG_irreg_part_mapping_snake_p_5_n_32} shows that only processor $P_0$ evaluates one more block than others. However, this mapping does not optimize communications insofar as a processor does not usually hold the upper blocks that are on the same row and column than the block that has been assigned to that processor.

\subsection{Our CGM-based parallel algorithm}

To solve the OBST problem, our CGM-based parallel algorithm evaluates the values of shortest paths to each node of a block, starting from the first diagonal of blocks to the diagonal $f(p) + k \times (\ceil*{f(p)/2} + 1)$, by using Knuth's sequential algorithm in computation rounds; although it does not enable to predict before starting the computation of the block $SM(i, j)$, which will be the values necessary for its evaluation. This is why the blocks are evaluated in a non-progressive way, i.e. the evaluation of blocks of the diagonal $d$ start after computing blocks of the diagonal $(d-1)$. 

The four-splitting technique can be adapted to evaluate the subblocks that are part of a block in a non-progressive fashion compared to the $k$-block splitting technique proposed in \citep{Lacmou2021}. Recall that when the latter is used, the number of subblocks (or $k$-blocks) drastically grows as the number of fragmentation rises. When a subblock is evaluated in  a non-progressive fashion by a processor, it must receive all the data of subblocks that it needs to start computations. However, since the subblocks are numerous and small, computing and communicating them to processors that need them will lead to communication overhead as a huge amount of communication must be done to exchange data \citep{Lacmou2021}. With the four-splitting technique, the blocks are subdivided into at most four, regardless of the number of fragmentation performed. This reduces the number of communication rounds and the number of steps to evaluate a whole block.

Denote the four subblocks of given block by $LL$, $LU$, $RL$, and $RU$, which respectively correspond to the subblock located in the leftmost lower corner, the leftmost upper corner, the rightmost lower corner, and the rightmost upper corner. Figures \ref{fig:4s_obst_step_0}, \ref{fig:4s_obst_step_1}, \ref{fig:4s_obst_step_2}, \ref{fig:4s_obst_step_3}, \ref{fig:4s_obst_step_4}, and \ref{fig:4s_obst_step_5} illustrate the different computation and communication steps of these subblocks by a processor :

\begin{enumerate}
	\item At step 0 in Figure \ref{fig:4s_obst_step_0}, no subblocks have started to be evaluated. This means that the processor is still receiving data of subblocks on which $LL$ depends.
	\item At step 1 in Figure \ref{fig:4s_obst_step_1}, $LL$ is computed but it is not communicated because a communication is not necessary at this step. Indeed, the processors which will receive $LL$ will not be able to start the computation as soon as possible because the evaluation is done in a non-progressive fashion. The ideal would be to wait to compute $LU$ before communicating them together.
	\item At step 2 in Figure \ref{fig:4s_obst_step_2}, the processor starts by receiving the lower subblocks that are in the same row than $LU$. Then, it computes $LU$. Finally, it communicates $LL$ and $LU$ to processors that need them. For example, the processor $P_5$ which evaluates the block of the second diagonal in Figure \ref{fig:DAG_irreg_part_mapping_snake_p_5_n_32} will send $LL$ and $LU$ to $P_1$ and $P_3$.
	\item At step 3 in Figure \ref{fig:4s_obst_step_3}, the processor receives the lower subblocks that are in the same column than $RL$. Then, it computes $RL$. Finally, it communicates $LL$ and $RL$ to processors that need them. By taking the previous example, $P_5$ will send $LL$ and $RL$ to $P_0$ and $P_2$.
	\item At step 4 in Figure \ref{fig:4s_obst_step_4}, the processor computes $RU$. Thereafter, it communicates $RL$ and $RU$ to processors that need them. By taking the previous example, $P_5$ will send $RL$ and $RU$ to $P_1$ and $P_4$.
	\item At step 5 in Figure \ref{fig:4s_obst_step_5}, all subblocks have been computed. The processor communicates $LU$ and $RU$ to processors that need them. By taking the previous example, $P_5$ will send $LU$ and $RU$ to $P_3$.
\end{enumerate}

Our CGM-based parallel algorithm based on the four-splitting technique to solve the OBST problem is given by Algorithm \ref{alg:obst_cgm_algorithm_4s_technique}.

\begin{figure}[!t]
	\centering
	\leavevmode
	\subfloat[][Step 0]{%
		\label{fig:4s_obst_step_0}
		\includegraphics[width=.313\columnwidth]{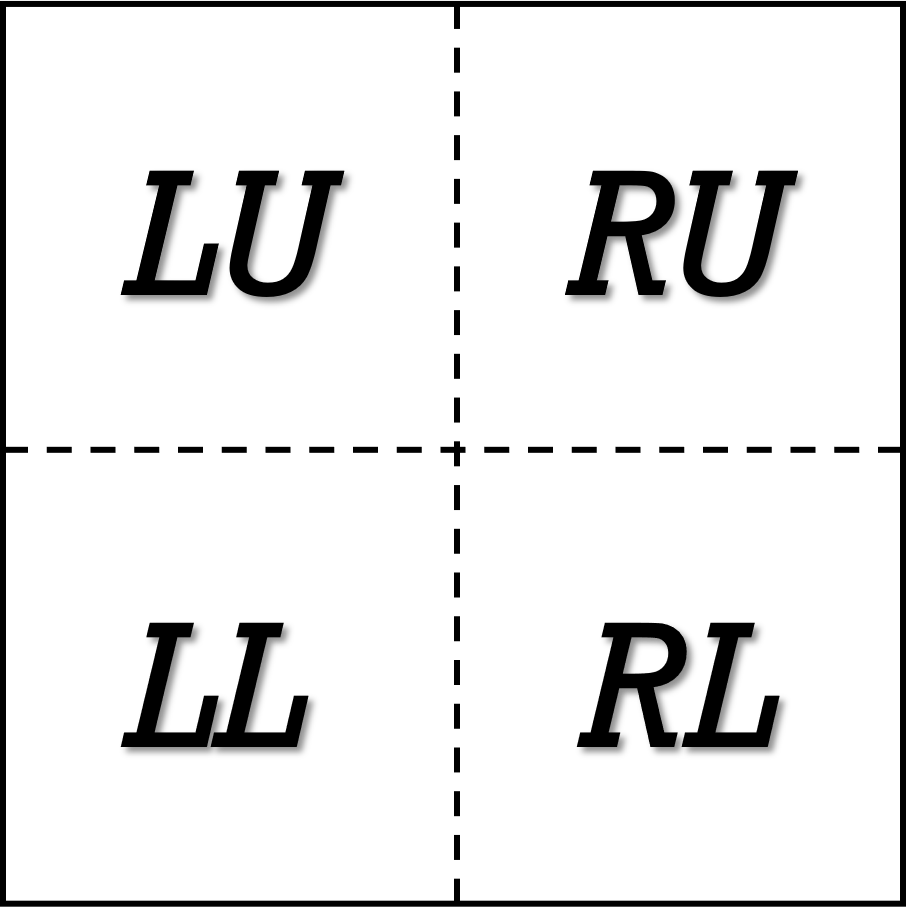}
	}
	\subfloat[][Step 1]{%
		\label{fig:4s_obst_step_1}
		\includegraphics[width=.313\columnwidth]{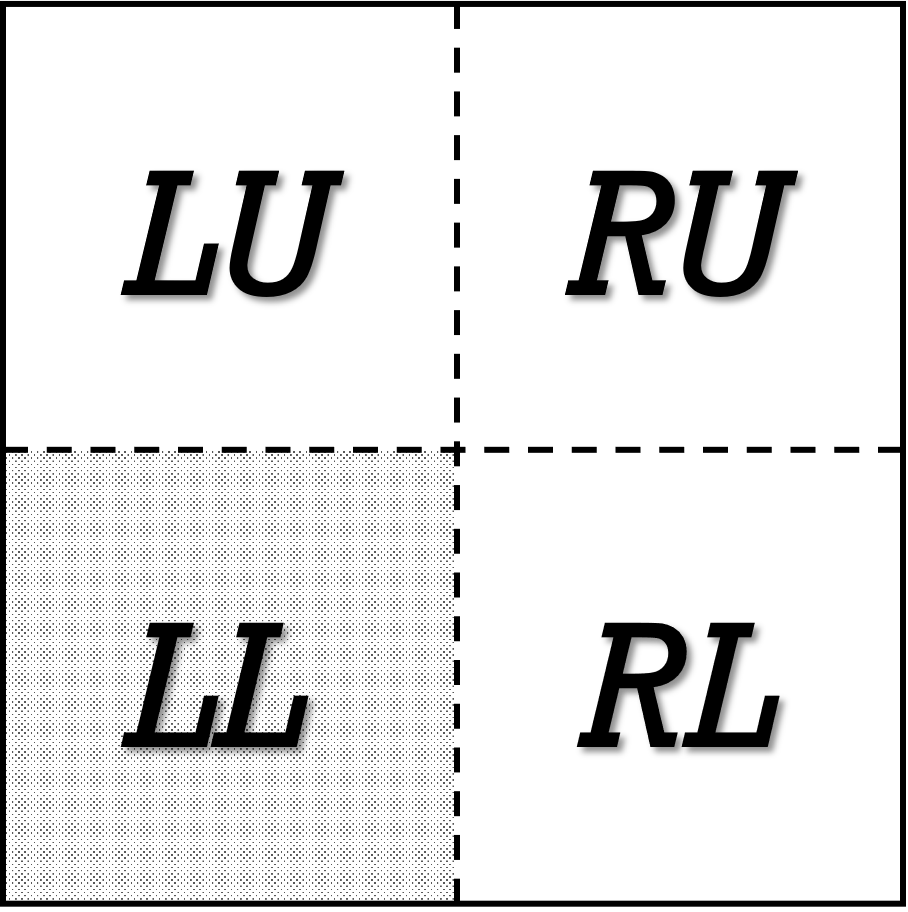}
	}
	\subfloat[][Step 2]{%
		\label{fig:4s_obst_step_2}
		\includegraphics[width=.313\columnwidth]{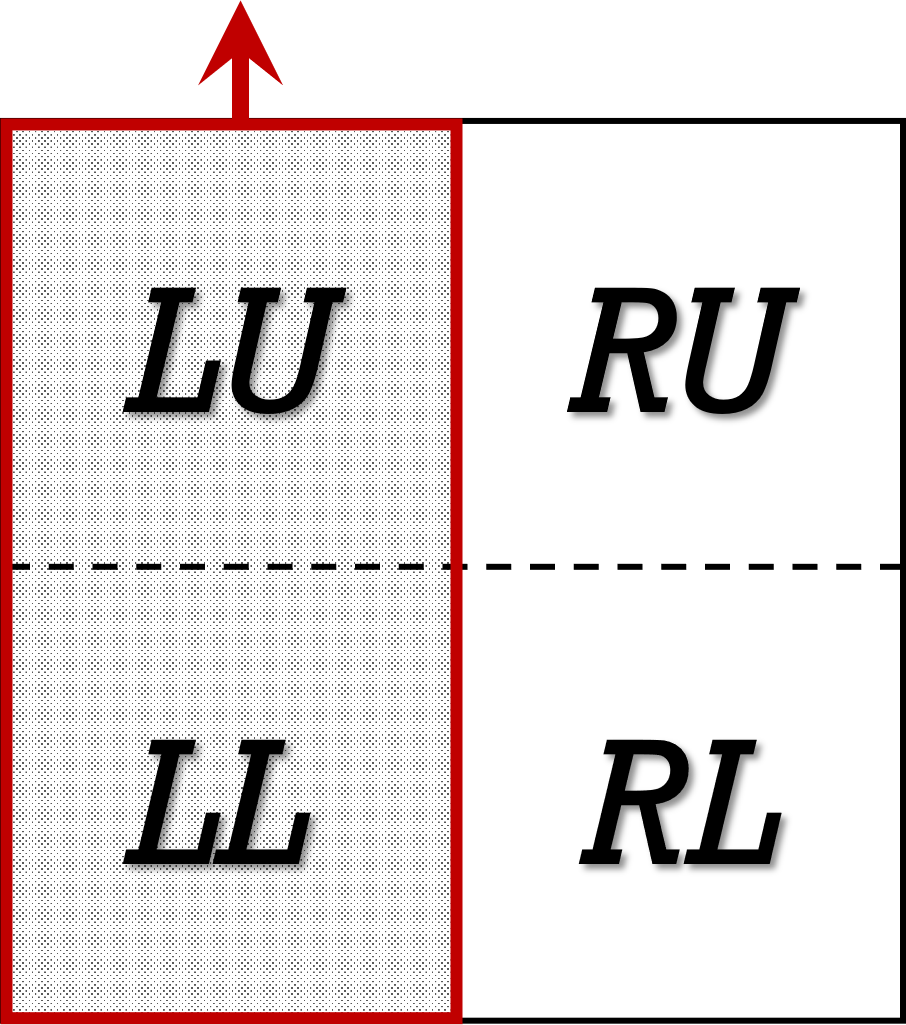}
	}
	\\
	\subfloat[][Step 3]{%
		\label{fig:4s_obst_step_3}
		\includegraphics[width=.325\columnwidth]{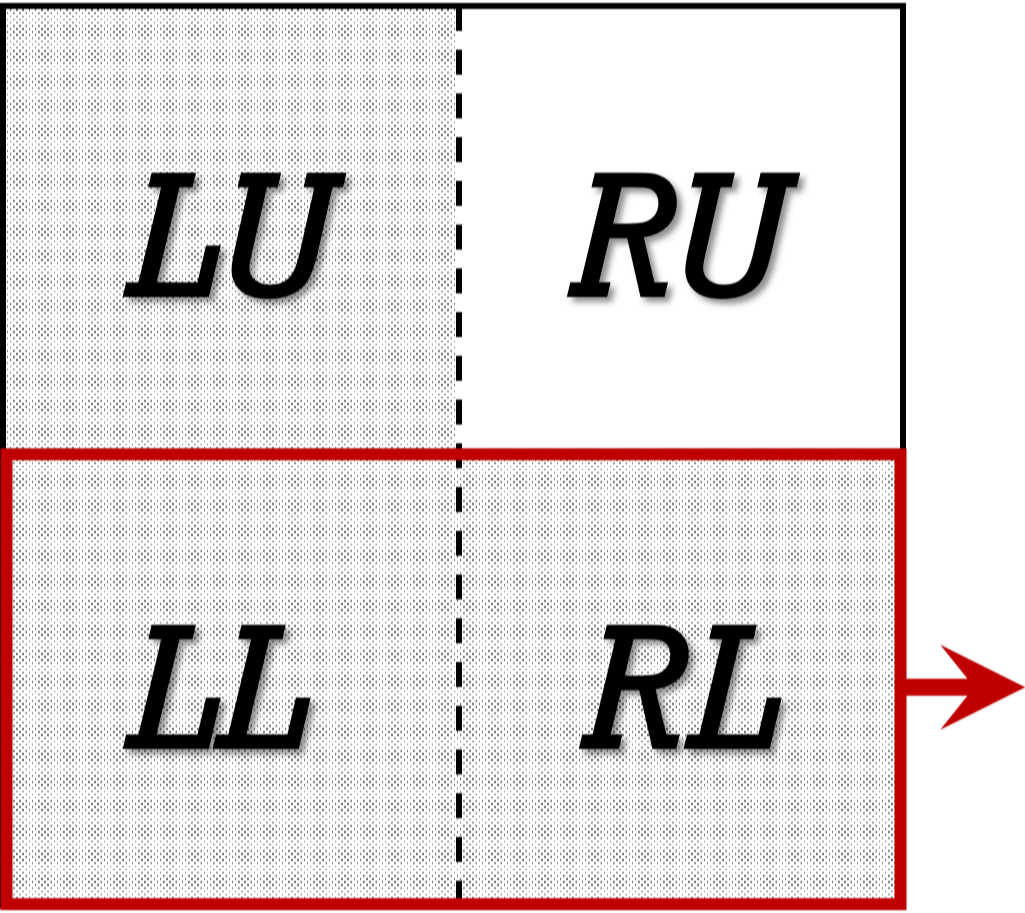}
	}
	\subfloat[][Step 4]{%
		\label{fig:4s_obst_step_4}
		\includegraphics[width=.289\columnwidth]{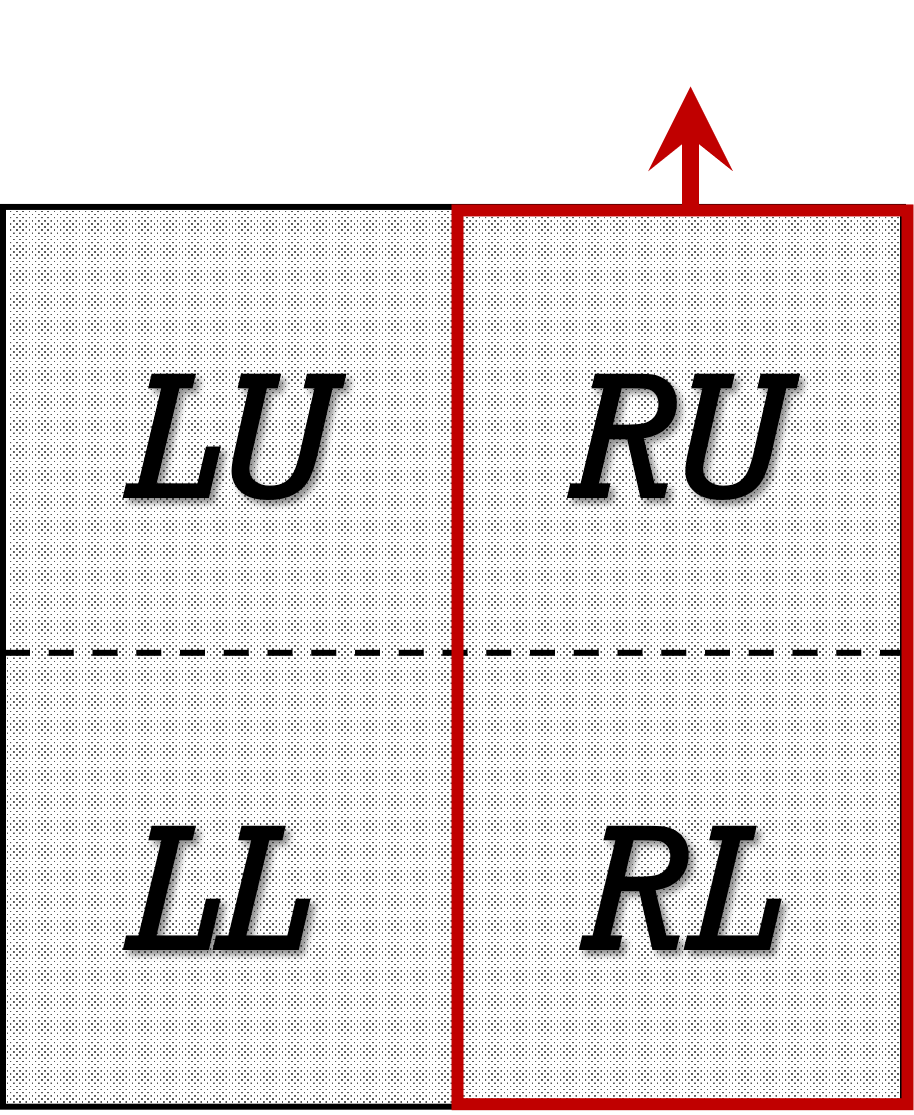}
	}
	\subfloat[][Step 5]{%
		\label{fig:4s_obst_step_5}
		\includegraphics[width=.325\columnwidth]{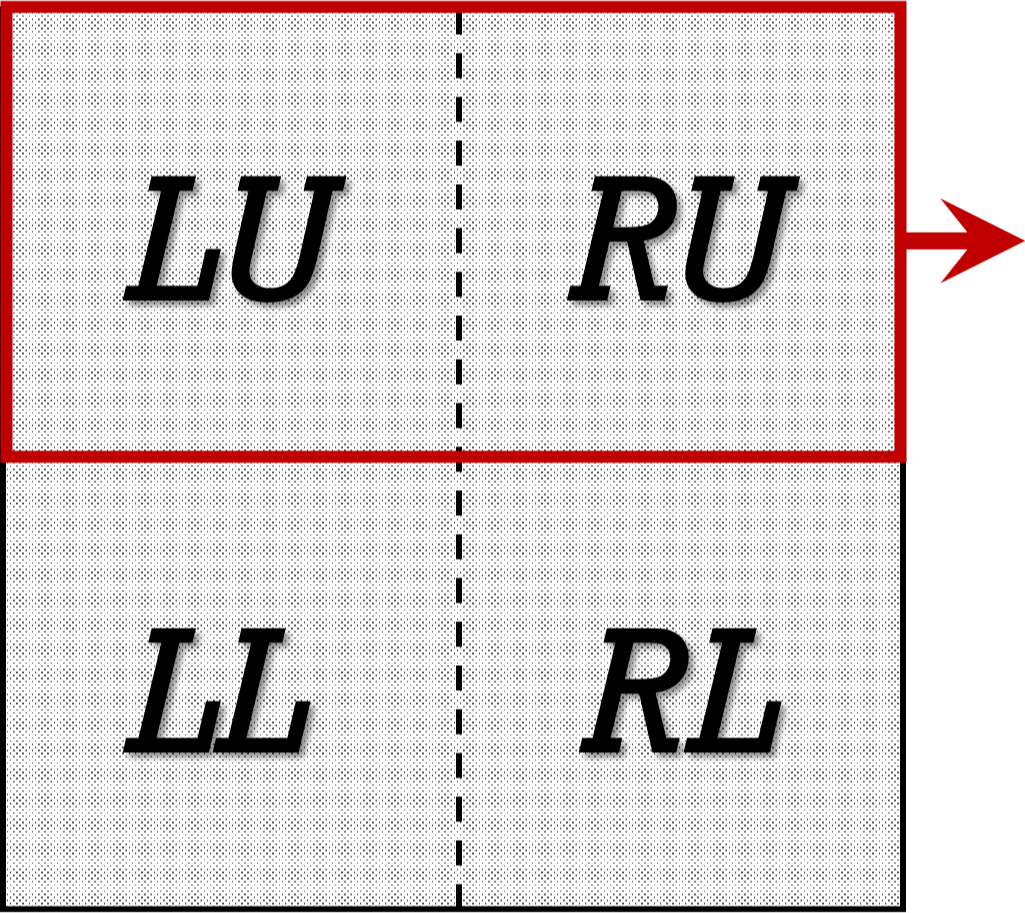}
	}
	\caption[Steps to evaluate the four subblocks of a given block when solving the OBST problem with the four-splitting technique]{Steps to evaluate the four subblocks of a given block when solving the OBST problem with the four-splitting technique. When a subblock is computed and needs to be communicated, it is colored in gray}
\end{figure}

\alglanguage{pseudocode}
\begin{algorithm}[!t]
	\caption{Our CGM-based parallel algorithm based on the four-splitting technique} \label{alg:obst_cgm_algorithm_4s_technique}
	\begin{algorithmic}[1]
		\For{$d = 1$ to $f(p) + k \times (\ceil*{f(p)/2} + 1)$}
		\State \textbf{Computation} of the subblocks $LL$ and $LU$ of blocks belonging to the round $d$ using Algorithm \ref{alg:obst_sequential_algorithm_knuth};
		\State \textbf{Communication} of entries ($Tree$ and $Cut$ tables) of $LL$ and $LU$ required for computing each block of rounds $\left\{d+1, d+2, \ldots, f(p) + k \times (\ceil*{f(p)/2} + 1)\right\}$;
		\State \textbf{Computation} of $RL$ using Algorithm \ref{alg:obst_sequential_algorithm_knuth};
		\State \textbf{Communication} of $LL$ and $RL$;
		\State \textbf{Computation} of $RU$ using Algorithm \ref{alg:obst_sequential_algorithm_knuth};
		\State \textbf{Communication} of $LU$, $RL$, and $RU$;
		\EndFor
	\end{algorithmic}
\end{algorithm}

\begin{theorem}
	Our CGM-based parallel solution based on the four-splitting technique requires $\mathcal{O}\left(n^2/\sqrt{p} \right)$ execution time with $\mathcal{O}\left( k \sqrt{p}\right)$ communication rounds in the worst case to solve the OBST problem. 
\end{theorem}

\begin{proof}
	Let $S = f(p)=\ceil*{\sqrt{2p}}$ and $\beta = (S \bmod 2)$. A single processor computes and communicates:
	
	\begin{enumerate}
		\item {three subblocks at the first diagonal of blocks;}
		\item {four subblocks from the diagonal of blocks 2 to $\floor*{S/2} - 1$;}
		\item {four subblocks for each $(\ceil*{S/2} + 1)$ diagonals of blocks belonging to the $l$th level of fragmentation such that $1 \le l < k$;}
		\item {one subblock for each $(S + \beta)$ diagonals of blocks belonging to the $k$th level of fragmentation.}
	\end{enumerate}
	During the computation rounds where Knuth's sequential algorithm is used to evaluate blocks, evaluating each subblock of a block belonging to the $l$th level of fragmentation requires $\mathcal{O}\left(\frac{n^2}{2^{2(l+1)} \times (2p)^{2/2}}\right) = \mathcal{O}\left(\frac{n^2}{4^{l+1} \times p}\right)$ local computation time. So, the evaluation of each diagonal required :
	\[
		\begin{split}
			D & = 3 \times \mathcal{O}\left(\frac{n^2}{4\times p}\right) + 4\left(\floor*{\frac{S}{2}} - 1\right) \times \mathcal{O}\left(\frac{n^2}{4\times p}\right) + 4\left(\ceil*{\frac{S}{2}}+1\right) \times \mathcal{O}\left(\frac{n^2}{4^2 \times p}\right) + \\
			& 
			4\left(\ceil*{\frac{S}{2}}+1\right) \times \mathcal{O}\left(\frac{n^2}{4^3 \times p}\right) + \cdots + 4\left(\ceil*{\frac{S}{2}}+1\right) \times \mathcal{O}\left(\frac{n^2}{4^k \times p}\right) + \\
			& \left(\ceil*{\frac{S}{2}}+1\right) \times \mathcal{O}\left(\frac{n^2}{4^k \times p}\right) + \left(S + \beta\right) \times \mathcal{O}\left(\frac{n^2}{4^k \times p}\right)\\
			& = \mathcal{O}\left(\frac{n^2}{\sqrt{p}}\right)
		\end{split}
	\]
	The number of communication rounds is equal to :
	\[
		\begin{split}
			E & = 3 + 4\left(\floor*{\frac{S}{2}} - 1\right) + 4 (k-1) \left(\ceil*{\frac{S}{2}} + 1\right) + S + \beta \\
			& = \mathcal{O}\left( k \sqrt{p}\right)
		\end{split}
	\]
	Therefore, this algorithm requires $\mathcal{O}\left(n^2/\sqrt{p}\right)$ execution time with $\mathcal{O}\left( k \sqrt{p}\right)$ communication rounds.
\end{proof}

\section{Experimental results}

\subsection{Experimental setups}

This section highlights the results obtained from experimentations of our CGM-based parallel solutions based on the four-splitting technique to solve the OBST problem, and compares them with the best previous solutions \citep{Kengne2019,Myoupo2014}. These experimentations have been performed on the cluster dolphin of the MatriCS platform of the University of Picardie Jules Verne\footnote{\url{https://www.u-picardie.fr/matrics}} and on our Raspberry Pi cluster (described hereunder). The MatriCS platform is composed of 48 nodes called thin nodes with 48 $\times$ 128GB of RAM, and 12 named thick nodes with 12 $\times$ 512GB of RAM. Each node is made of two Intel Xeon Processor E5-2680 V4 (35M Cache, 2.40 GHz), and each of them consists of 14 cores. All nodes are interconnected with OmniPath links providing 100Gbps throughput. These algorithms have been implemented in the C programming language\footnote{The source codes are available here : \url{https://github.com/compiii/CGM-Sol-for-OBST}.}, and on the operating system CentOS Linux release 7.6.1810. The MPI library (OpenMPI version 1.10.4) has been used for inter-processor communication. These algorithms have been executed on five thin nodes. The results are presented following the different values of $(n, p, k)$, where :
\begin{itemize}
	\item $n$ is the data size, with values in the set $\langle$4096, 8192, 12288, 16384, 20480, 24576, 28672, 32768, 36864, 40960$\rangle$.
	\item $p$ is the number of processors, with values in the set $\langle$1, 32, 64, 96, 128$\rangle$. When $p = 1$, the sequential algorithm of Knuth \cite{Knuth1971} is carried out.
	\item $k$ is the number of fragmentations performed, with values in the set $\langle$0, 1, 2, 3, 4, 5$\rangle$. When $k = 0$, our solution is similar to the one in \citep{Myoupo2014}.
\end{itemize}
Tables \ref{tab:chap:3:obst_total_exec_time_for_p_32_four_splitting_technique_matrics}, \ref{tab:chap:3:obst_speedup_efficiency_for_p_32_four_splitting_technique_matrics}, \ref{tab:chap:3:obst_total_exec_time_for_n_36864_40960_four_splitting_technique_matrics}, \ref{tab:chap:3:obst_speedup_efficiency_for_n_36864_40960_four_splitting_technique_matrics}, and \ref{tab:chap:3:obst_total_exec_time_speedup_for_p_32_four_splitting_technique_compiii}
show the total execution time, the speedup, and the efficiency of our CGM-based parallel solutions to solve the OBST problem. In these tables, the solution based on the irregular partitioning technique is referred to \textit{frag} and the one based on the four-splitting technique is referred to \textit{4s}. Figures \ref{fig:chap:3:obst_evol_4s_com_time_p_32_compared_to_frag}, \ref{fig:chap:3:obst_evol_4s_com_time_n_40960_compared_to_frag}, \ref{fig:chap:3:obst_evol_4s_time_calc_com_idle_exec_p_32}, \ref{fig:chap:3:obst_evol_4s_time_calc_com_idle_exec_n_40960}, \ref{fig:chap:3:obst_evol_4s_load_balancing_k_0_and_1}, \ref{fig:chap:3:obst_evol_4s_load_balancing_k_0_and_2}, \ref{fig:chap:3:obst_evol_4s_exec_time_p_32_compared_to_frag}, \ref{fig:chap:3:obst_evol_4s_exec_time_n_40960_compared_to_frag}, \ref{fig:chap:3:obst_evol_4s_time_calc_com_idle_exec_p_32_matrics}, and \ref{fig:chap:3:obst_evol_4s_time_calc_com_idle_exec_compiii} are drawn from the results obtained in Tables \ref{tab:chap:3:obst_total_exec_time_for_p_32_four_splitting_technique_matrics}, \ref{tab:chap:3:obst_speedup_efficiency_for_p_32_four_splitting_technique_matrics}, \ref{tab:chap:3:obst_total_exec_time_for_n_36864_40960_four_splitting_technique_matrics}, \ref{tab:chap:3:obst_speedup_efficiency_for_n_36864_40960_four_splitting_technique_matrics}, and \ref{tab:chap:3:obst_total_exec_time_speedup_for_p_32_four_splitting_technique_compiii}.

\begin{sidewaystable}[!t]
	\centering
	\caption{Total execution time (in seconds) for $n \in \{4096, \ldots, 40960 \}$, $p = 32$, and $k \in \{0, \ldots, 5\}$ on the MatriCS platform}\label{tab:chap:3:obst_total_exec_time_for_p_32_four_splitting_technique_matrics}
	\begin{tabular}{|l|r|r|rr|rrrrr|}
		\cline{4-10}
		\multicolumn{3}{c|}{}  & \multicolumn{2}{c|}{\textit{frag}} & \multicolumn{5}{c|}{\textit{4s}}\\\hline
		$n$   & $p = 1$  & $k = 0$  & $k=1$    & $k=2$    & $k=1$ 	 & $k=2$ 	& $k=3$   & $k=4$ 	& $k=5$ \\\hline
		4096  & 62.34    & 20.58    & 16.68    & 13.11    & 11.50    & 9.90     & 8.94    & 8.58    & 8.52    \\
		8192  & 514.32   & 180.82   & 136.43   & 100.02   & 97.34    & 77.30    & 70.76   & 69.10   & 68.64   \\
		12288 & 1784.19  & 621.52   & 476.46   & 343.27   & 338.27   & 266.48   & 243.43  & 236.72  & 235.48  \\
		16384 & 4315.43  & 1582.57  & 1142.21  & 818.11   & 813.47   & 634.41   & 585.65  & 572.21  & 568.37  \\
		20480 & 8546.86  & 2919.91  & 2308.19  & 1600.75  & 1605.26  & 1244.97  & 1147.37 & 1123.66 & 1116.53 \\
		24576 & 14899.60 & 5970.20  & 3920.49  & 2785.83  & 2793.95  & 2151.68  & 1977.78 & 1937.93 & 1930.08 \\
		28672 & 23802.78 & 8527.82  & 6078.21  & 4428.91  & 4458.84  & 3444.42  & 3155.25 & 3089.17 & 3076.31 \\
		32768 & 33263.71 & 12533.26 & 10130.78 & 6607.36  & 6682.07  & 5153.27  & 4718.72 & 4616.36 & 4587.60 \\
		36864 & 48356.46 & 17658.48 & 13514.64 & 9451.97  & 9578.16  & 7353.35  & 6729.28 & 6572.38 & 6529.41 \\
		40960 & 69722.12 & 24910.11 & 20933.65 & 12986.77 & 13219.62 & 10120.25 & 9252.37 & 9013.49 & 8961.54 \\\hline
	\end{tabular}
	
		\caption{Speedup and efficiency (in \%) for $n \in \{4096, \ldots, 40960 \}$, $p = 32$, and $k \in \{1, \ldots, 5\}$ on the MatriCS platform}\label{tab:chap:3:obst_speedup_efficiency_for_p_32_four_splitting_technique_matrics}
	\begin{tabular}{|l|rr|rrrrr|rr|rrrrr|}
		\cline{2-15}
		\multicolumn{1}{c|}{}& \multicolumn{7}{c|}{Speedup}  & \multicolumn{7}{c|}{Efficiency}\\
		\cline{2-15}
		\multicolumn{1}{c|}{}  & \multicolumn{2}{c|}{\textit{frag}} & \multicolumn{5}{c|}{\textit{4s}}& \multicolumn{2}{c|}{\textit{frag}} & \multicolumn{5}{c|}{\textit{4s}}\\\hline
		\diagbox[height=0.45cm]{$n$}{$k$}   & 1    & 2    & 1 	 & 2 	& 3   & 4 	& 5 & 1    & 2    & 1 	 & 2 	& 3   & 4 	& 5 \\\hline
		4096  & 3.74  & 4.75 & 5.42  & 6.29  & 6.97  & 7.27  & 7.32 & 11.68 & 14.85 & 16.94 & 19.67 & 21.79 & 22.71 & 22.87 \\
		8192  & 3.77  & 5.14 & 5.28  & 6.65  & 7.27  & 7.44  & 7.49 & 11.78 & 16.07 & 16.51 & 20.79 & 22.71 & 23.26 & 23.42 \\
		12288 & 3.74  & 5.20 & 5.27  & 6.70  & 7.33  & 7.54  & 7.58 & 11.70 & 16.24 & 16.48 & 20.92 & 22.90 & 23.55 & 23.68 \\
		16384 & 3.78  & 5.27 & 5.30  & 6.80  & 7.37  & 7.54  & 7.59 & 11.81 & 16.48 & 16.58 & 21.26 & 23.03 & 23.57 & 23.73 \\
		20480 & 3.70  & 5.34 & 5.32  & 6.87  & 7.45  & 7.61  & 7.65 & 11.57 & 16.69 & 16.64 & 21.45 & 23.28 & 23.77 & 23.92 \\
		24576 & 3.80  & 5.35 & 5.33  & 6.92  & 7.53  & 7.69  & 7.72 & 11.88 & 16.71 & 16.67 & 21.64 & 23.54 & 24.03 & 24.12 \\
		28672 & 3.92  & 5.37 & 5.34  & 6.91  & 7.54  & 7.71  & 7.74 & 12.24 & 16.80 & 16.68 & 21.60 & 23.57 & 24.08 & 24.18 \\
		32768 & 3.28  & 5.03 & 4.98  & 6.45  & 7.05  & 7.21  & 7.25 & 10.26 & 15.73 & 15.56 & 20.17 & 22.03 & 22.52 & 22.66 \\
		36864 & 3.58  & 5.12 & 5.05  & 6.58  & 7.19  & 7.36  & 7.41 & 11.18 & 15.99 & 15.78 & 20.55 & 22.46 & 22.99 & 23.14 \\
		40960 & 3.33  & 5.37 & 5.27  & 6.89  & 7.54  & 7.74  & 7.78 & 10.41 & 16.78 & 16.48 & 21.53 & 23.55 & 24.17 & 24.31 \\\hline
	\end{tabular}
\end{sidewaystable}

\begin{sidewaystable}[!t]
	\centering
			\caption{Total execution time (in seconds) for $n \in \{ 36864, 40960 \}$, $p \in \{64, 96, 128\}$, and $k \in \{ 0, \ldots, 5\}$ on the MatriCS platform}\label{tab:chap:3:obst_total_exec_time_for_n_36864_40960_four_splitting_technique_matrics}
	\begin{tabular}{|l|r|rr|rrrrr|}
		\cline{3-9}
		\multicolumn{2}{c|}{}  & \multicolumn{2}{c|}{\textit{frag}} & \multicolumn{5}{c|}{\textit{4s}}\\\hline
		\diagbox[height=0.45cm]{$p$}{$k$}   & 0  & 1    & 2    & 1 	 & 2 	& 3   & 4 	& 5 \\\hline
		\multicolumn{9}{|c|}{$n = 36864$} \\\hline
		64  & 11671.40 & 7610.84  & 6586.99  & 6849.69 & 5268.74 & 4744.92 & 4609.23 & 4492.70  \\
		96  & 10128.02 & 6573.25  & 5552.72  & 5811.06 & 4336.53 & 3960.81 & 3863.48 & 3846.97  \\
		128 & 8951.40  & 5788.29  & 5126.66  & 5166.80 & 3858.00 & 3488.67 & 3407.10 & 3391.69  \\\hline
		\multicolumn{9}{|c|}{$n = 40960$} \\\hline
		64  & 16045.47 & 10456.06 & 9057.29 & 9564.80 & 7256.84 & 6333.78 & 6208.79 & 6177.57 \\
		96  & 13928.60 & 9031.15  & 7813.92 & 7980.85 & 5954.35 & 5415.29 & 5299.00 & 5276.19 \\
		128 & 12305.95 & 7949.95  & 6707.94 & 7115.22 & 5316.16 & 4815.67 & 4688.91 & 4670.79 \\\hline
	\end{tabular}
	
	\caption{Speedup and efficiency (in \%) for $n \in \{ 36864, 40960 \}$, $p \in \{64, 96, 128\}$, and $k \in \{ 1, \ldots, 5\}$ on the MatriCS platform}\label{tab:chap:3:obst_speedup_efficiency_for_n_36864_40960_four_splitting_technique_matrics}
	\begin{tabular}{|l|rr|rrrrr|rr|rrrrr|}
		\cline{2-15}
		\multicolumn{1}{c|}{}& \multicolumn{7}{c|}{Speedup}  & \multicolumn{7}{c|}{Efficiency}\\
		\cline{2-15}
		\multicolumn{1}{c|}{}  & \multicolumn{2}{c|}{\textit{frag}} & \multicolumn{5}{c|}{\textit{4s}}& \multicolumn{2}{c|}{\textit{frag}} & \multicolumn{5}{c|}{\textit{4s}}\\\hline
		\diagbox[height=0.45cm]{$p$}{$k$}   & 1    & 2    & 1 	 & 2 	& 3   & 4 	& 5 & 1    & 2    & 1 	 & 2 	& 3   & 4 	& 5 \\\hline
		\multicolumn{15}{|c|}{$n = 36864$} \\\hline
		64  & 6.35  & 7.34 & 7.06  & 9.18  & 10.19 & 10.49 & 10.76 & 9.93  & 11.47 & 11.03 & 14.34 & 15.92 & 16.39 & 16.82 \\
		96  & 7.36  & 8.71 & 8.32  & 11.15 & 12.21 & 12.52 & 12.57 & 7.66  & 9.07  & 8.67  & 11.62 & 12.72 & 13.04 & 13.09 \\
		128 & 8.35  & 9.43 & 9.36  & 12.53 & 13.86 & 14.19 & 14.26 & 6.53  & 7.37  & 7.31  & 9.79  & 10.83 & 11.09 & 11.14 \\\hline
		\multicolumn{15}{|c|}{$n = 40960$} \\\hline
		64  & 6.67  & 7.70  & 7.29  & 9.61  & 11.01 & 11.23 & 11.29 & 10.42 & 12.03 & 11.39 & 15.01 & 17.20 & 17.55 & 17.63 \\
		96  & 7.72  & 8.92  & 8.74  & 11.71 & 12.88 & 13.16 & 13.21 & 8.04  & 9.29  & 9.10  & 12.20 & 13.41 & 13.71 & 13.77 \\
		128 & 8.77  & 10.39 & 9.80  & 13.12 & 14.48 & 14.87 & 14.93 & 6.85  & 8.12  & 7.66  & 10.25 & 11.31 & 11.62 & 11.66 \\\hline
	\end{tabular}
	
		\caption{Total execution time (in seconds) and speedup for $n \in \{4096, \ldots, 16384 \}$, $p \in \{1,32\}$, and $k \in \{1, \ldots, 5\}$ while solving the OBST problem with the four-splitting technique on our Raspberry Pi cluster}\label{tab:chap:3:obst_total_exec_time_speedup_for_p_32_four_splitting_technique_compiii}
	\begin{tabular}{|l|rrrrrr|rrrrr|}
		\cline{2-12}
		\multicolumn{1}{c|}{} & \multicolumn{6}{c|}{Total execution time} & \multicolumn{5}{c|}{Speedup}\\\hline
		$n$   & $p = 1$  & $k=1$   & $k=2$   & $k=3$   & $k=4$   & $k=5$   & $k=1$ & $k=2$ & $k=3$ & $k=4$ & $k=5$ \\\hline
		4096  & 152.45   & 70.53   & 61.74   & 50.45   & 51.68   & 52.24   & 2.16  & 2.47  & 3.02  & 2.95  & 2.92  \\
		8192  & 1252.78  & 553.87  & 425.86  & 415.23  & 385.80  & 391.02  & 2.26  & 2.94  & 3.02  & 3.25  & 3.20  \\
		12288 & 4384.93  & 2670.32 & 1615.90 & 1448.08 & 1318.49 & 1351.73 & 1.64  & 2.71  & 3.03  & 3.33  & 3.24  \\
		16384 & 10660.09 & 4856.32 & 4034.42 & 3870.51 & 3610.49 & 3464.09 & 2.20  & 2.64  & 2.75  & 2.95  & 3.08  \\\hline
	\end{tabular}
\end{sidewaystable}

\subsection{Evolution of the global communication time}
Figures \ref{fig:chap:3:obst_evol_4s_com_time_p_32_compared_to_frag} and \ref{fig:chap:3:obst_evol_4s_com_time_n_40960_compared_to_frag} depict the global communication time of \emph{frag} and \emph{4s} while solving the OBST problem by performing one and two fragmentations. Recall that the global communication time is composed of the latency time of processors and the effective transfer time of data. These figures show that the global communication time of \emph{frag} and \emph{4s} is relatively lower than the one based on the regular partitioning (when $k = 0$). It gradually decreases as the number of fragmentations increases. For example, on thirty-two processors when $n = 40960$, the global communication time of \emph{frag} (respectively \emph{4s}) decreases on average by 3.87\% (respectively 41.84\%) when $k = 1$ and on average by 40.49\% (respectively 54.83\%) when $k = 2$. Indeed, when a fragmentation is performed, the number of blocks increase by dividing the current size of blocks into four to form the smaller-size blocks. It enables to minimize the latency time of processors and the effective transfer time of data since the smaller-size blocks take less time to evaluate and communicate compared to the larger-size blocks.

\begin{figure}[!t]
	\centering
	\leavevmode
	\subfloat[][$p = 32$]{%
		\label{fig:chap:3:obst_evol_4s_com_time_p_32_compared_to_frag}
		\includegraphics[width=.485\columnwidth]{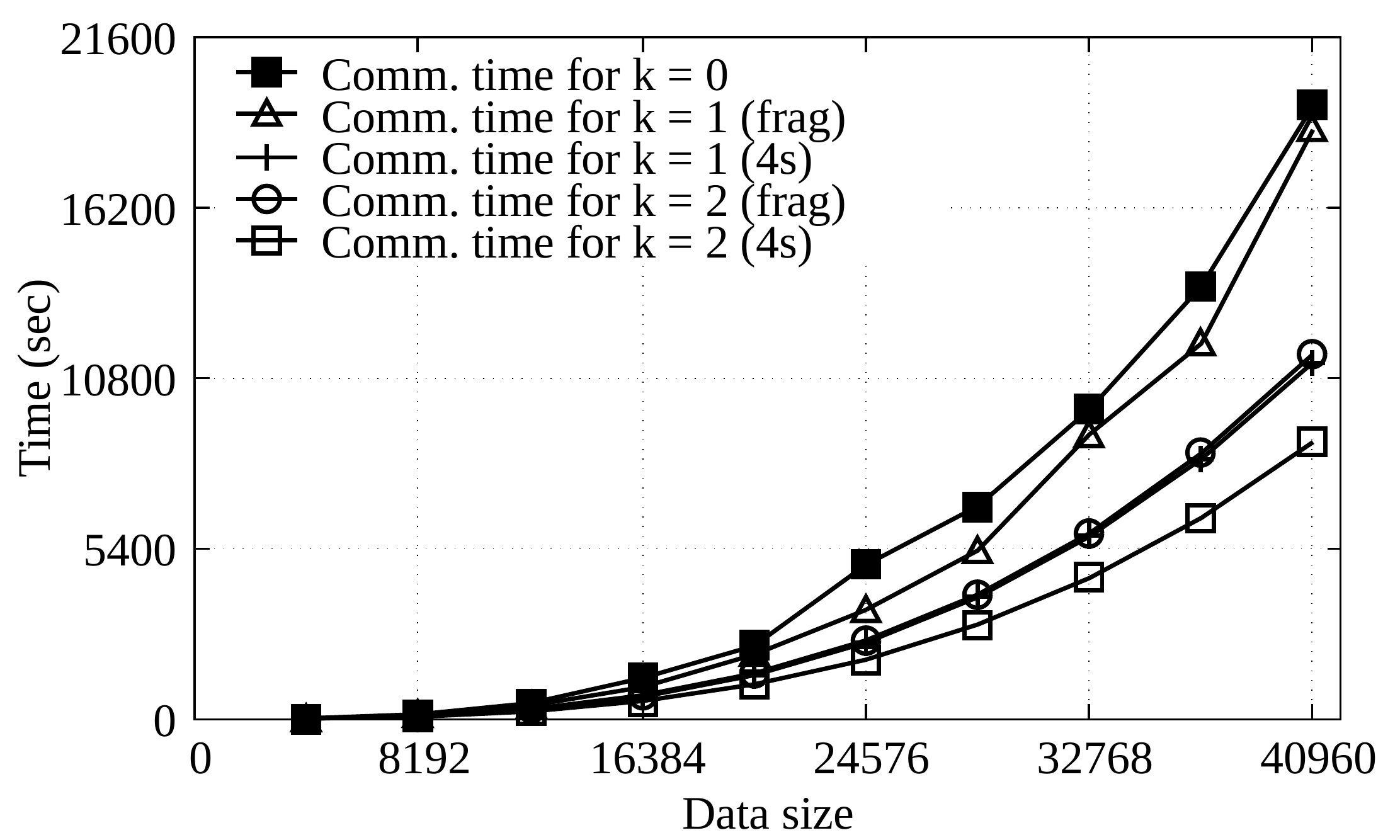}
	}
	\subfloat[][$n = 40960$]{%
		\label{fig:chap:3:obst_evol_4s_com_time_n_40960_compared_to_frag}
		\includegraphics[width=.485\columnwidth]{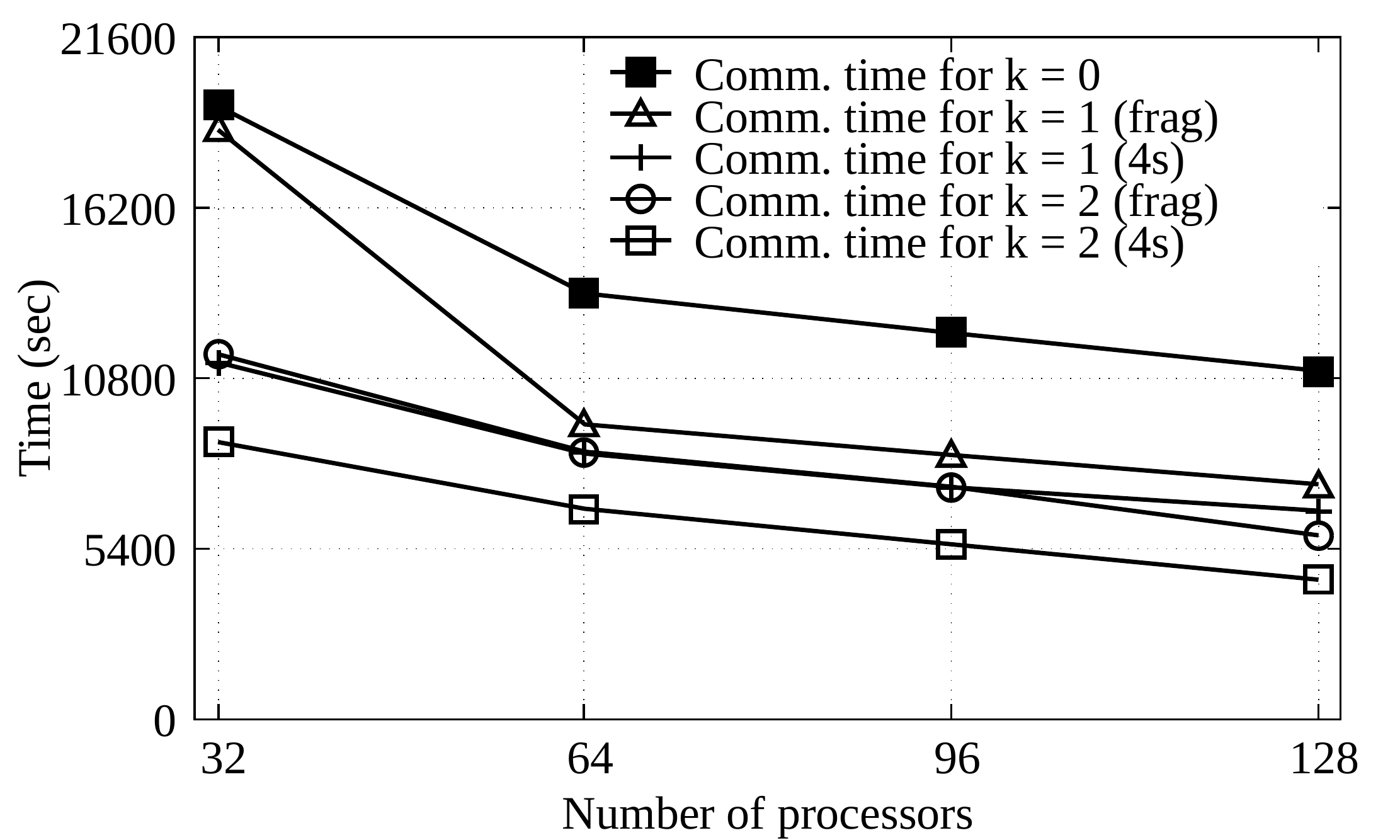}
	}
	\caption{Global communication time for $n \in \{4096, \ldots, 40960 \}$, $p \in \{32,\ldots,128\}$, and $k \in \{ 0, 1, 2\}$ on the MatriCS platform}
\end{figure}

Figures \ref{fig:chap:3:obst_evol_4s_com_time_p_32_compared_to_frag} and \ref{fig:chap:3:obst_evol_4s_com_time_n_40960_compared_to_frag} also show that the global communication time of \emph{4s} is lower than \emph{frag}. It decreases significantly when the number of fragmentations increases. These results were expected because \emph{4s} minimizes the latency time of processors by enabling them to start evaluating subblocks as soon as the data they need are available. Grouping the subblocks into two before communicating also contributes to this minimization because it limits the number of messages to be exchanged in the network.

Figures \ref{fig:chap:3:obst_evol_4s_time_calc_com_idle_exec_p_32} and \ref{fig:chap:3:obst_evol_4s_time_calc_com_idle_exec_n_40960} show that the global communication time and the overall computation time of \emph{4s} gradually decrease while solving the OBST problem by successively performing five fragmentations. It can be noticed that until the fifth fragmentation the global communication time does not deteriorate. For example, on thirty-two processors when $n = 40960$, it decreases in average by 41.84\% when $k = 1$, in average by 54.83\% when $k = 2$, in average by 58.18\% when $k = 3$, in average by 59.43\% when $k = 4$, and in average by 69.08\% when $k = 5$. This is due to the sequential algorithm of Knuth \citep{Knuth1971} which is used to evaluate small subblocks that are close to the optimal solution. Since they do not require a high evaluation time, processors that need these subblocks will not have to wait long to receive them.

\begin{figure*}[!t]
	\centering
	\leavevmode
	\subfloat[][$p = 32$]{%
		\label{fig:chap:3:obst_evol_4s_time_calc_com_idle_exec_p_32}
		\includegraphics[width=.485\columnwidth]{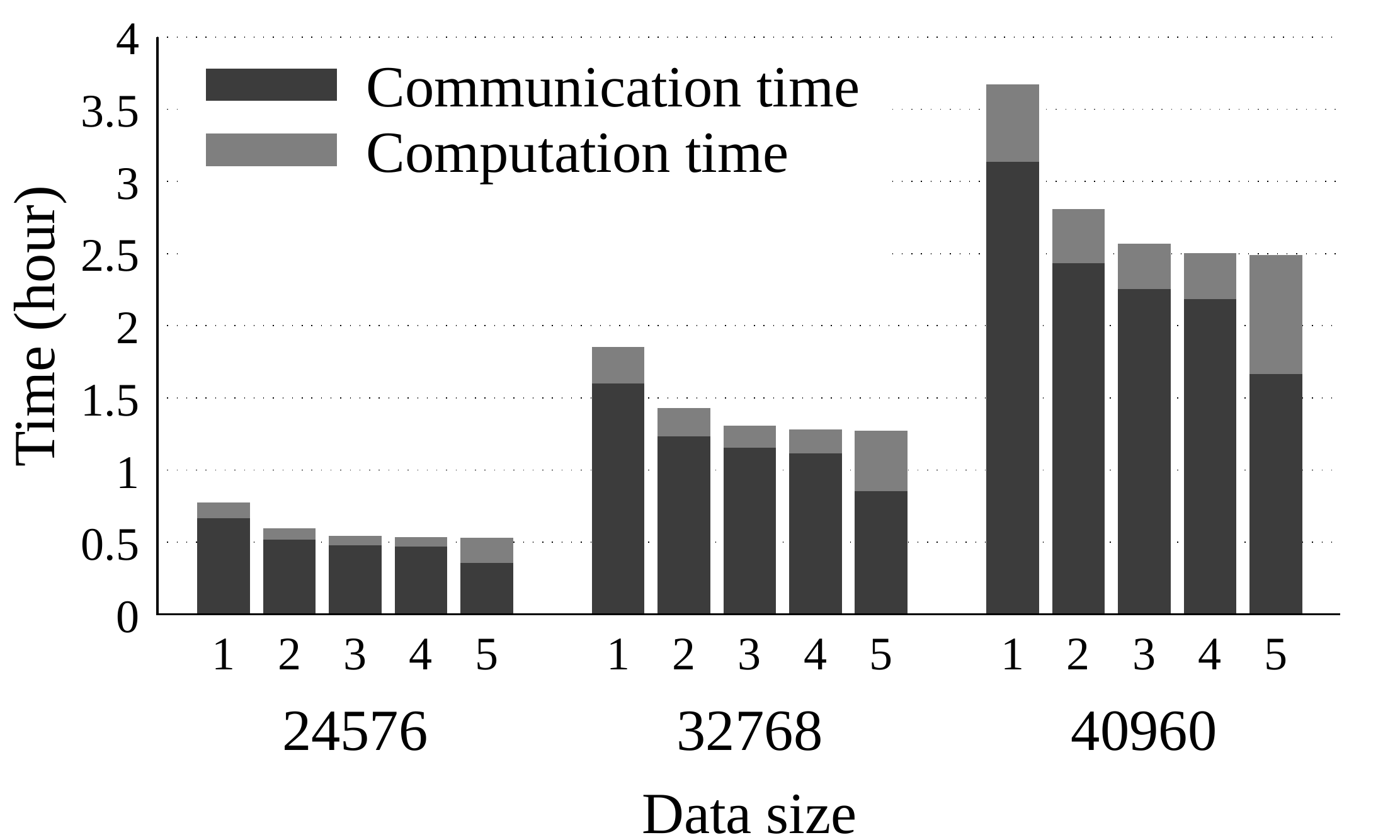}
	}
	\subfloat[][$n = 40960$]{%
		\label{fig:chap:3:obst_evol_4s_time_calc_com_idle_exec_n_40960}
		\includegraphics[width=.485\columnwidth]{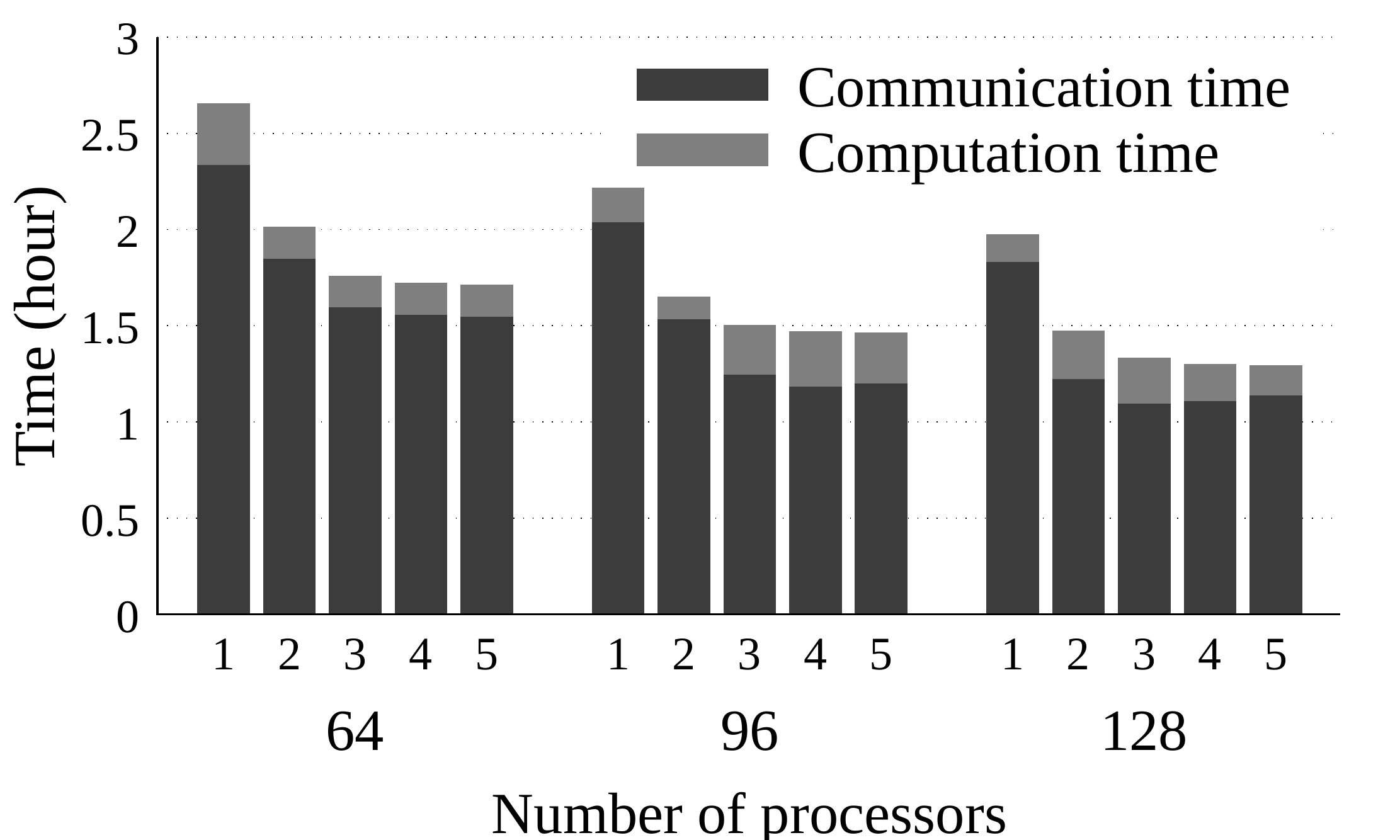}
	}
	\caption{Comparison of the overall computation time and the global communication time for $n \in \{24576, 32768, 40960 \}$, $p \in \{32,\ldots,128\}$, and $k \in \{ 1, \ldots, 5\}$ while solving the OBST problem with the four-splitting technique on the MatriCS platform}
\end{figure*}

\subsection{Evolution of the load-balancing of processors}
Figures \ref{fig:chap:3:obst_evol_4s_load_balancing_k_0_and_1} and \ref{fig:chap:3:obst_evol_4s_load_balancing_k_0_and_2} compare the load imbalance on thirty-two processors by performing the difference between the average computation time of processors and the lowest (and highest) computation time among them. The processor with the lowest or highest computational load varies with the number of fragmentations performed. For example, when $k = 0$ and $k = 2$, respectively, $P_7$ and $P_{30}$ have the lowest loads, and $P_3$ and $P_{22}$ have the highest.

These figures show that the irregular partitioning technique of the dynamic graph balances the load of processors better than the regular partitioning technique of this graph due to the gradual reduction of the block sizes that enables processors to remain active as long as possible. Indeed, the fragmentation is performed on the blocks with the largest loads. So, some small-size blocks (which are in the upper diagonals) have higher loads than the large ones (which are in the lower diagonals). Thanks to the snake-like mapping which enables  to distribute blocks onto processors in an equitable way, a processor can have in the worst case one more block than another. This greatly contributes to balancing the loads of processors. When $n = 40960$ and $k = 1$, the lowest load narrows down to 27.52\% for \emph{frag} and to 27.24\% for \emph{4s}; and the highest load increases up to 11.89\% for \emph{frag} and to 11.60\% for \emph{4s}. In the same vein when $k = 2$, the lowest load and the highest load decrease on average by 56.63\% and 42.98\% for \emph{frag}, and on average by 55.73\% and 42.77\% for \emph{4s}. Therefore, our CGM-based parallel solution based on the four-splitting technique (\emph{4s}) kept the good performance of our CGM-based parallel solution based on the irregular partitioning technique (\emph{frag}) with respect to the load-balancing of processors since they promote the load balancing when the number of fragmentations increases. These results were predictable since \emph{4s} is based on the irregular partitioning technique and the snake-like mapping.

\begin{figure*}[!t]
	\centering
	\leavevmode
	\subfloat[][$k = 1$]{%
		\label{fig:chap:3:obst_evol_4s_load_balancing_k_0_and_1}
		\includegraphics[width=.485\columnwidth]{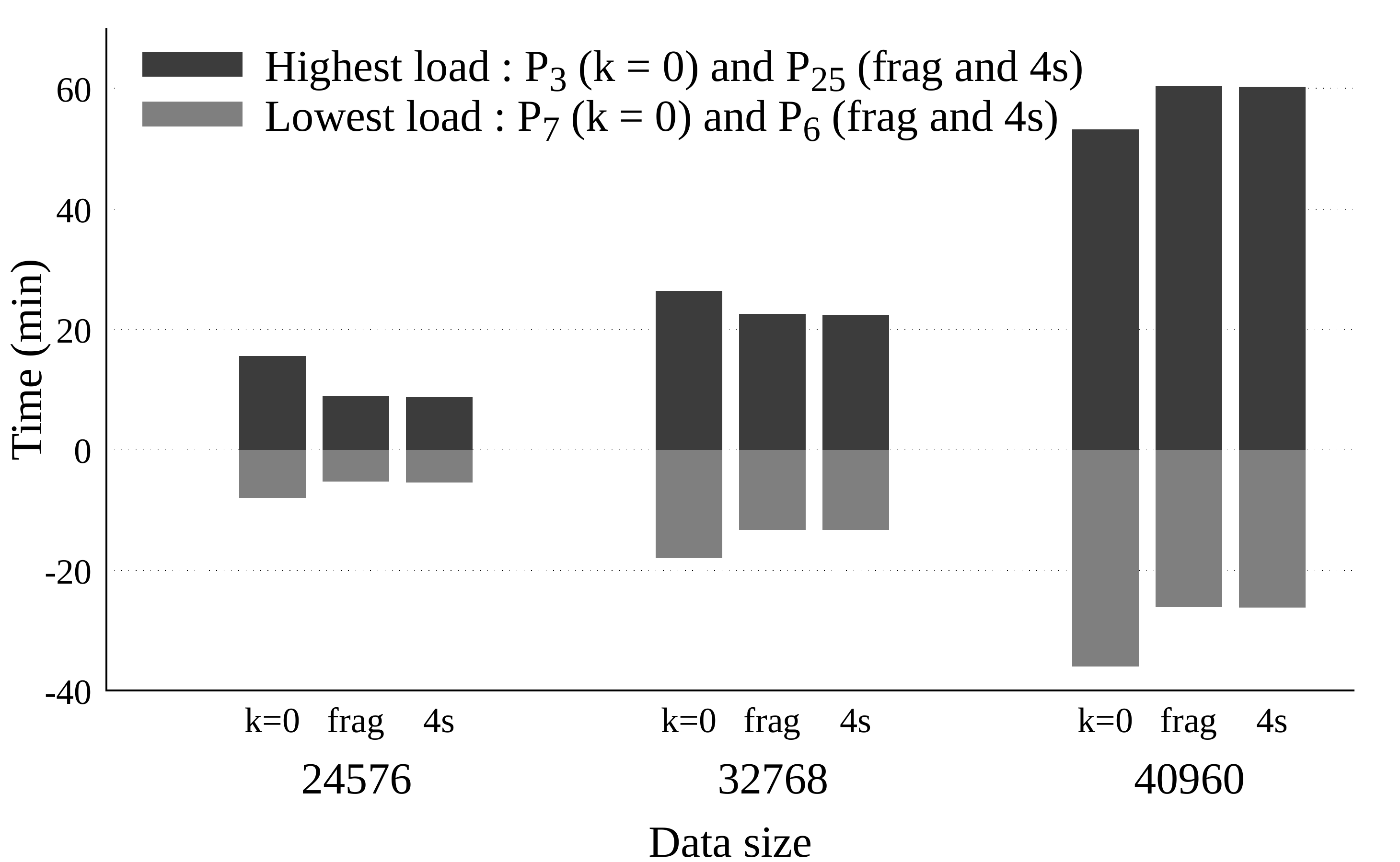}
	}
	\subfloat[][$k = 2$]{%
		\label{fig:chap:3:obst_evol_4s_load_balancing_k_0_and_2}
		\includegraphics[width=.485\columnwidth]{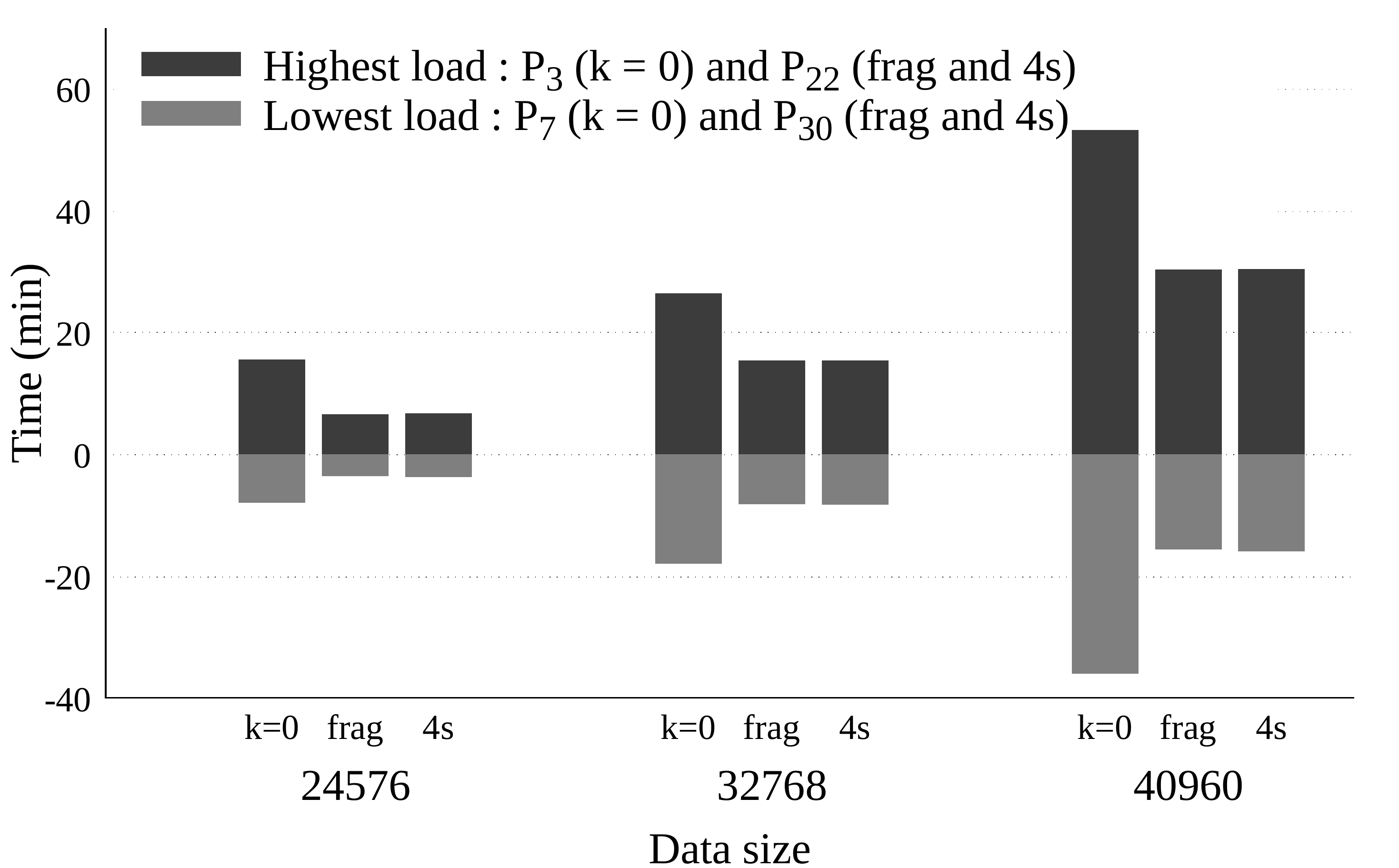}
	}
	\caption{Load imbalance of processors for $n \in \{ 24576, 32768, 40960 \}$, $p = 32$, and $k \in \{ 0, 1, 2\}$ on the MatriCS platform}
\end{figure*}

\subsection{Evolution of the total execution time}
\begin{figure*}[!t]
	\centering
	\leavevmode
	\subfloat[][$p = 32$]{%
		\label{fig:chap:3:obst_evol_4s_exec_time_p_32_compared_to_frag}
		\includegraphics[width=.485\columnwidth]{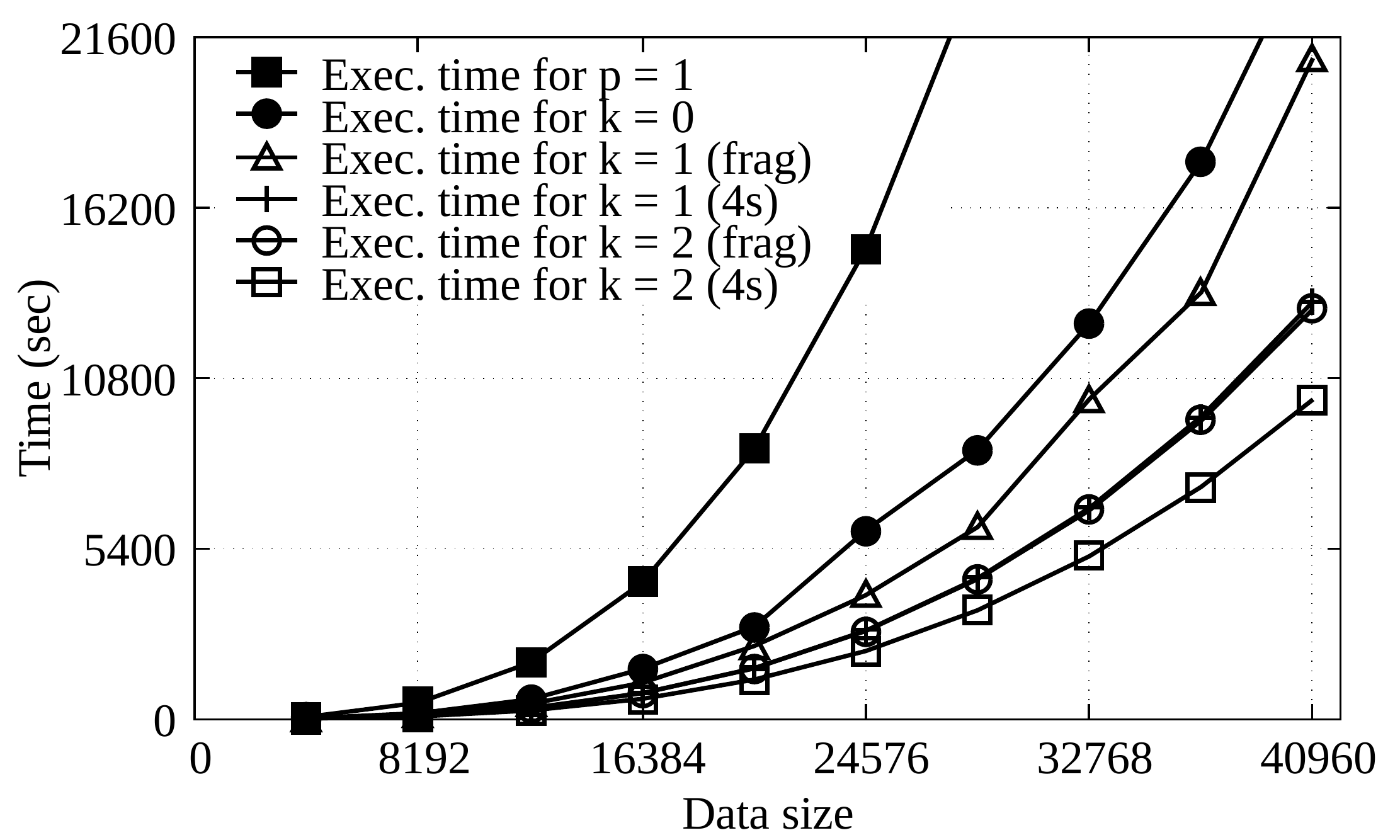}
	}
	\subfloat[][$n = 40960$]{%
		\label{fig:chap:3:obst_evol_4s_exec_time_n_40960_compared_to_frag}
		\includegraphics[width=.485\columnwidth]{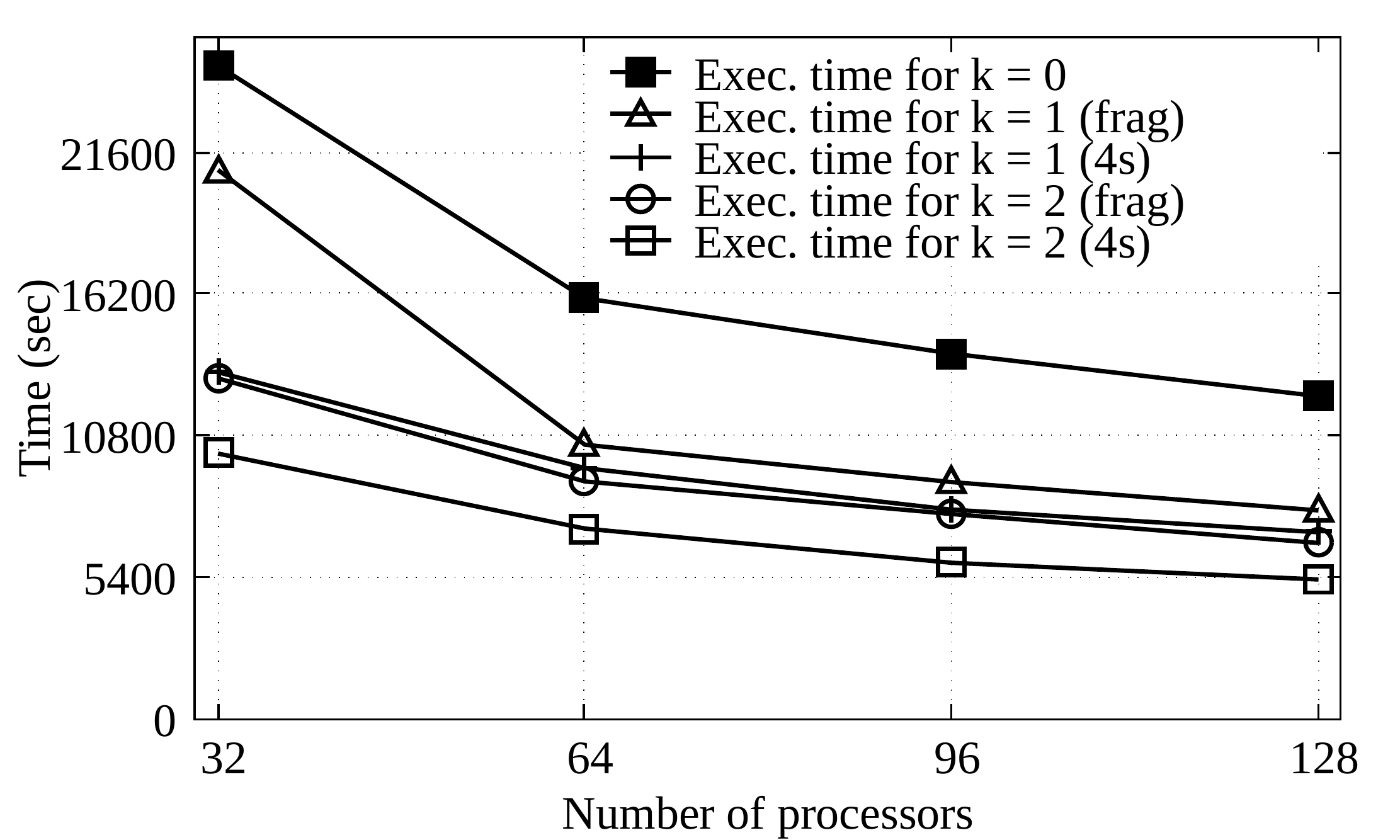}
	}
	\caption{Total execution time for $n \in \{4096, \ldots, 40960 \}$, $p \in \{32,\ldots,128\}$, and $k \in \{ 0, 1, 2\}$ on the MatriCS platform}
\end{figure*}
Figures \ref{fig:chap:3:obst_evol_4s_time_calc_com_idle_exec_p_32}, \ref{fig:chap:3:obst_evol_4s_time_calc_com_idle_exec_n_40960}, \ref{fig:chap:3:obst_evol_4s_exec_time_p_32_compared_to_frag}, and \ref{fig:chap:3:obst_evol_4s_exec_time_n_40960_compared_to_frag} illustrate the total execution time of \emph{frag} and \emph{4s} while solving the OBST problem by performing one, two, three, four, and five fragmentations. These figures show that the minimization of the global communication time (due to the minimization of the latency time of processors) lead up to a reduction of the total execution time as the number of fragmentations rises. For example, on thirty-two processors when $n = 40960$ and $k = 1$ (respectively $k = 2$) in Figure \ref{fig:chap:3:obst_evol_4s_time_calc_com_idle_exec_p_32}, the total execution time decreases on average by 15.96\% (respectively 47.87\%) for \emph{frag} and on average by 46.93\% (respectively 59.37\%) for \emph{4s}. Similar
observations can be made on one hundred and twenty-eight processors in Figure \ref{fig:chap:3:obst_evol_4s_time_calc_com_idle_exec_n_40960}; where when $n = 40960$ and $k = 1$ (respectively $k = 2$), the total execution time decreases on average by 35.40\% (respectively  45.50\%) for \emph{frag} and on average by 42.18\% (respectively 56.80\%) for \emph{4s}. It is then obvious to notice the huge impact that the global communication time has on the total execution time, and consequently on the speedup and the efficiency (as shown in Tables \ref{tab:chap:3:obst_total_exec_time_for_p_32_four_splitting_technique_matrics}, \ref{tab:chap:3:obst_speedup_efficiency_for_p_32_four_splitting_technique_matrics}, \ref{tab:chap:3:obst_total_exec_time_for_n_36864_40960_four_splitting_technique_matrics}, and \ref{tab:chap:3:obst_speedup_efficiency_for_n_36864_40960_four_splitting_technique_matrics}). The following results can be noticed : 

\begin{itemize}
	\item From $n = 4096$ to 40960 on thirty-two to one hundred and twenty-eight processors, the total execution time of \emph{4s} is better than that of \emph{frag} when performing one and two fragmentations. For example, on thirty-two processors when $n = 40960$ and $k = 1$ (respectively $k = 2$), the speedup and the efficiency are equal to 3.33 and 10.41\% (respectively 5.37 and 16.78\%) for \emph{frag}, and increase up to 5.27 and 16.48\% (respectively 6.89 and 21.53\%) for \emph{4s}.
	\item When performing three, four, and five additional fragmentations, the total execution time of \emph{4s} keeps good performance, although from the fourth fragmentation the performance gain is not meaningful anymore. For example, on thirty-two processors when $n = 40960$, the speedup and the efficiency are equal to 7.54 and 23.55\% when $k = 3$, and rise up to 7.74 and 24.17\% when $k = 4$, up to 7.78 and 24.31\% when $k = 5$.
\end{itemize}
From all this, we can deduce that it is better to use our CGM-based parallel solution based on the four-splitting technique to solve the OBST problem since it outperforms the one using the irregular partitioning technique and it is scalable as the data size, the number of processors, and the number of fragmentations rise.

\subsection{What is the ideal number of fragmentations ?}
This is a legitimate question because as shown earlier, on the MatriCS platform the total execution time of \emph{4s} degrades from the fourth fragmentation while solving the MPP but keeps good performance until the fifth fragmentation while solving the OBST problem. One idea to determine the ideal number of fragmentations would be to run our CGM-based parallel solution based on the four-splitting technique to solve the OBST problem for example on another cluster. If we obtain the same result that was archived on the MatriCS platform then this could guide us to determine this number. Otherwise, we will deduce that the ideal number of fragmentations depends on the platform where the parallel algorithm is executed.

To achieve this goal, we have built a cluster consisting of four compute nodes. Each node is a Raspberry Pi 4 computer model B composed of a 64-bit quad-core Cortex-A72 processor with 8GB of RAM and an operating system Raspbian GNU/Linux 11. All nodes are interconnected with Cat8 Ethernet cable through a TP-Link router (model TL-SG108E) providing 16Gbps throughput and the inter-processor communication has been ensured by OpenMPI version 4.1.2.

Figures \ref{fig:chap:3:obst_evol_4s_time_calc_com_idle_exec_p_32_matrics} and \ref{fig:chap:3:obst_evol_4s_time_calc_com_idle_exec_compiii} compare the overall computation time and global communication time of \emph{4s} on thirty-two processors while solving the OBST problem by performing one, two, three, four, and five fragmentations on the MatriCS platform and on our Raspberry Pi cluster. It is obvious to notice that whatever the cluster, \emph{4s} progressively reduces the global communication time when the number of fragmentations rises. However on the Raspberry Pi cluster, the performance of \emph{4s} degrades when adding the overall computation time at the fifth fragmentation for $n = 8192$ and $n = 12288$ because the global communication time did not reduce enough to keep the good performance. For example when $n = 12288$ in Table \ref{tab:chap:3:obst_total_exec_time_speedup_for_p_32_four_splitting_technique_compiii}, the speedup is equal to 1.64 when $k = 1$, and increases up to 2.71 when $k = 2$, up to 3.03 when $k = 3$, up to 3.33 when $k = 4$, but the speedup narrows down to 3.24 when $k = 5$. This result is due to the fact that the Raspberry Pi cluster does not have a communication network as fast as the MatriCS platform. Indeed, an additional fragmentation is no longer necessary from the fourth fragmentation. Thus, since the network is not very fast, excessive communication caused a drop in performance. In summary, we deduce that the ideal number of fragmentations depends on the architecture where the CGM-based parallel solution is executed.

Another relevant observation can be seen on the Raspberry Pi cluster in Figure \ref{fig:chap:3:obst_evol_4s_time_calc_com_idle_exec_compiii}, where the total execution time of \emph{4s} gradually decreases after five successive fragmentations when $n = 16384$. As shown in Table \ref{tab:chap:3:obst_total_exec_time_speedup_for_p_32_four_splitting_technique_compiii}, the speedup is equal to 2.20 when $k = 1$, and increases up to 2.64 when $k = 2$, up to 2.75 when $k = 3$, up to 2.95 when $k = 4$, and up to 3.08 when $k = 5$. This is due to the fact that when the input data size increases, fragmentations are more and more required until the performance gain is not significant or drops. However, as seen earlier, the performance of \emph{4s} degrades at the fifth fragmentation for $n = 8192$ and $n = 12288$. Consequently, we infer that the ideal number of fragmentations also depends on the input data size.

\begin{figure}[!t]
	\centering
	\leavevmode
	\subfloat[][MatriCS platform]{%
		\label{fig:chap:3:obst_evol_4s_time_calc_com_idle_exec_p_32_matrics}
		\includegraphics[width=.485\columnwidth]{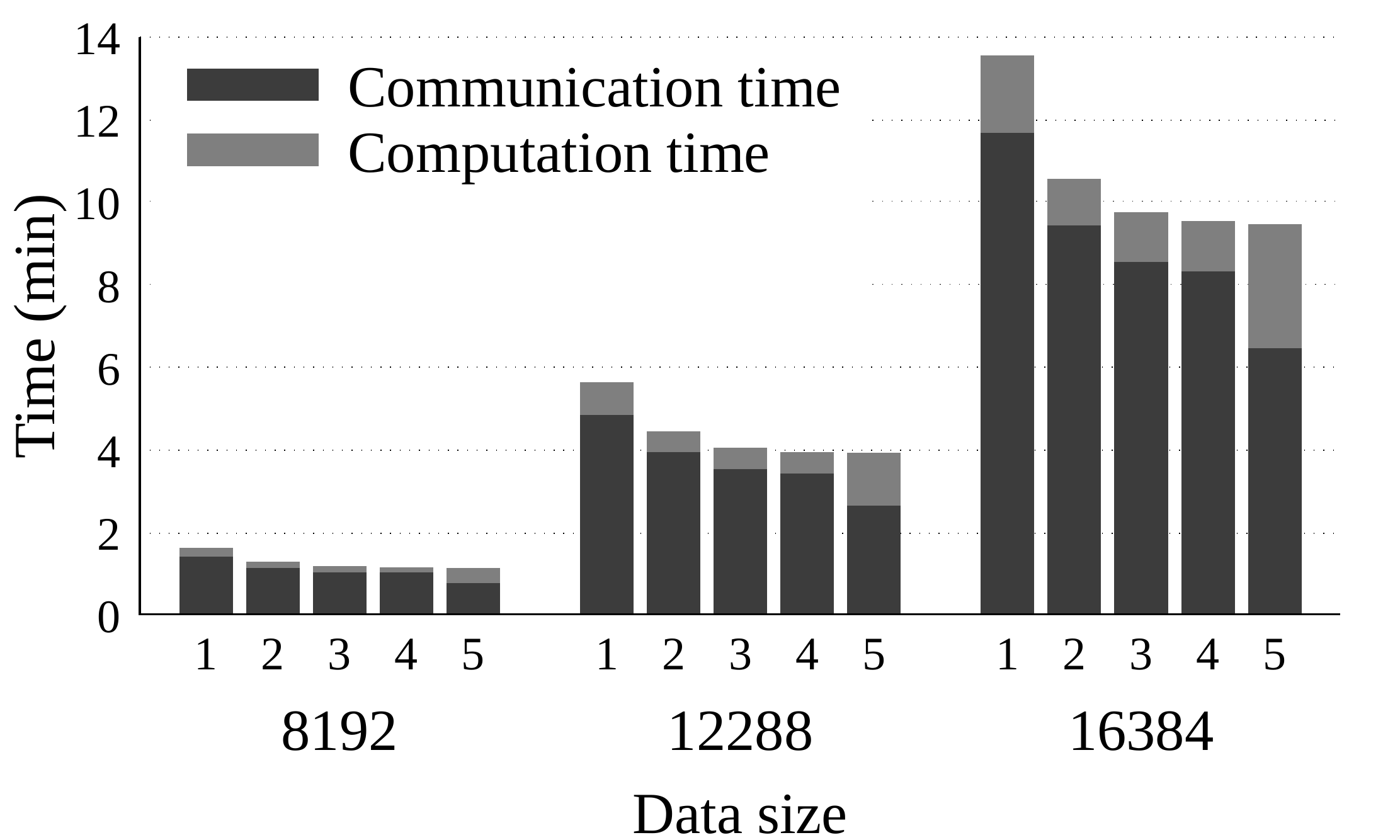}
	}
	\subfloat[][Raspberry Pi cluster]{%
		\label{fig:chap:3:obst_evol_4s_time_calc_com_idle_exec_compiii}
		\includegraphics[width=.485\columnwidth]{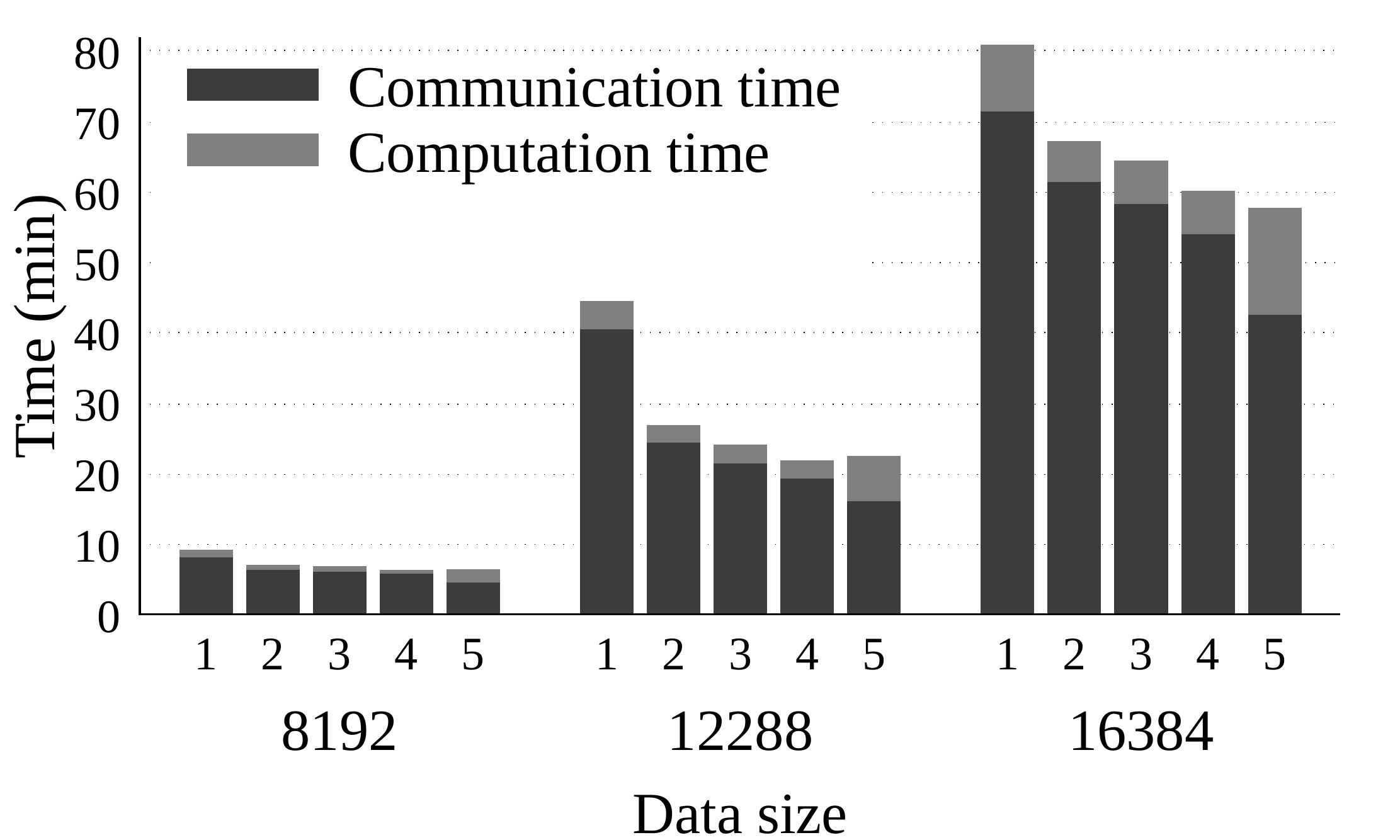}
	}
	\caption{Comparison of the overall computation time and the global communication time for $n \in \{8192, 12288, 16384 \}$, $p = 32$, and $k \in \{1, \ldots, 5\}$ while solving the OBST problem with the four-splitting technique on on the MatriCS platform and our Raspberry Pi cluster}
\end{figure}

\section{Conclusion}
We presented in this paper a CGM-based parallel solution to solve the optimal binary search tree problem using the four-splitting technique. It is a partitioning technique consisting of subdividing the dynamic graph into subgraphs (or blocks) of variable size and splitting large-size blocks into four subblocks. It enables processors to start the evaluation of blocks as soon as possible to minimize their latency time, which is the largest part of the global communication time. Our CGM-based parallel solution requires $\mathcal{O}\left(n^2/\sqrt{p} \right)$ execution time with $\mathcal{O}\left( k \sqrt{p}\right)$ communication rounds. Experimental results showed a good agreement with theoretical predictions since it significantly decreases the global communication time compared to the solution using the irregular partitioning technique \cite{Kengne2019} (and thus minimizes the total execution time). It also minimizes the number of communication rounds and balances the load of processors as well as the irregular partitioning technique. We also conducted experimentations to determine the ideal number of times the size of the blocks must be subdivided. The results showed that this number depends on the input data size and the architecture where the solution is executed. It would be interesting to undertake the challenge to determine this number before starting or during computations since it depends on the architecture where the solution is executed. It would be also interesting to apply our partitioning techniques on shared-memory architectures and GPU architectures.

\section*{Disclosure statement}
The authors declare no conflicts of interest associated with this manuscript.

\section*{Data availability statement}
The authors confirm that the data supporting the findings of this work are available within the manuscript.


\bibliographystyle{tfnlm}
\bibliography{refs}

\end{document}